\documentclass[a4paper,10pt]{article}
\usepackage{charter}
\usepackage[usenames,dvipsnames]{color}
\usepackage[pdftex,breaklinks,colorlinks,citecolor={BlueViolet},linkcolor={Blue},urlcolor=Maroon]{hyperref}
\usepackage{float}
\usepackage{comment}
\usepackage{cite}
\usepackage{fullpage}
\usepackage{graphicx}
\usepackage{mathrsfs,amssymb,amsmath,amsthm}
\usepackage{enumerate}
\usepackage{listings}
\usepackage{tikz}
\usepackage[ruled]{algorithm2e}

\definecolor{dkgreen}{rgb}{0,0.6,0}
\definecolor{gray}{rgb}{0.5,0.5,0.5}
\definecolor{mauve}{rgb}{0.58,0,0.82}

\newtheorem{theorem}{Theorem}
\newtheorem{lemma}[theorem]{Lemma}

\newtheorem{definition}[theorem]{Definition}

\newtheorem{proposition}[theorem]{Proposition}

\newcommand{\newloglike}[2]{\newcommand{#1}{\mathop{\rm #2}\nolimits}}
\newloglike{\sgn}{sgn}

\newcommand{\abs}[1]{\left\vert#1\right\vert}
\newcommand{\set}[1]{\left\{#1\right\}}
\newcommand{\tuple}[1]{\left(#1\right)} \newcommand{\eps}{\varepsilon}
\newcommand{\tp}{\tuple}
\newcommand{\sharpP}{\#\mathbf{P}} \renewcommand{\P}{\mathbf{P}}
\newcommand{\defeq}{\triangleq}
\def\*#1{\mathbf{#1}}
\def\+#1{\mathcal{#1}}

 \renewcommand{\mid}{\;\middle\vert\;}

\usepackage{xifthen}

\renewcommand{\Pr}[2][]{ \ifthenelse{\isempty{#1}}
  {\mathbf{Pr}\left[#2\right]} {\mathbf{Pr}_{#1}\left[#2\right]} }
\newcommand{\E}[2][]{ \ifthenelse{\isempty{#1}}
  {\mathop{\mathbf{E}}\left[#2\right]}
  {\mathop{\mathbf{E}}_{#1}\left[#2\right]} }

\newcommand{\nul}[1]{{\it et al.\/}}

\newcommand{\FPTAS}{\textup{\textbf{FPTAS}}~}

\newcommand{\pderiv}[2]{\frac{\partial{#1}}{\partial{#2}}}

\usepackage{todonotes}

\renewcommand{\deg}[2][]{\mathrm{deg}_{#1}\tuple{#2}}

\begin{document}

\title{An FPTAS for Counting Proper Four-Colorings on Cubic Graphs}
\author{ Pinyan Lu \thanks{Institute for Theoretical Computer Science,
    School of Information Management and Engineering, Shanghai
    University of Finance and
    Economics. \texttt{lu.pinyan@mail.shufe.edu.cn}} \and Kuan Yang
  \thanks{Department of Computer Science, University of Oxford. {\tt
      kuan.yang.6@gmail.com}} \and Chihao Zhang \thanks{Institute of
    Theoretical Computer Science and Communications, The Chinese
    University of Hong Kong. {\tt chihao.zhang@gmail.com}} \and
  Minshen Zhu \thanks{Purdue University. {\tt minshen.zh@gmail.com}} }
\date{}

\maketitle

\begin{abstract}
  Graph coloring is arguably the most exhaustively studied problem in
  the area of approximate counting. It is conjectured that there is a
  fully polynomial-time (randomized) approximation scheme
  (FPTAS/FPRAS) for counting the number of proper colorings as long as
  $q \geq \Delta + 1$, where $q$ is the number of colors and $\Delta$
  is the maximum degree of the graph. The bound of $q = \Delta + 1$ is
  the uniqueness threshold for Gibbs measure on $\Delta$-regular
  infinite trees.  However, the conjecture remained open even for any
  fixed $\Delta\geq 3$ (The cases of $\Delta=1, 2$ are trivial). In
  this paper, we design an FPTAS for counting the number of proper
  $4$-colorings on graphs with maximum degree $3$ and thus confirm the
  conjecture in the case of $\Delta=3$.  This is the first time to
  achieve this optimal bound of $q = \Delta + 1$. Previously, the best
  FPRAS requires $q > \frac{11}{6} \Delta$ and the best deterministic
  FPTAS requires $q > 2.581\Delta + 1$ for general graphs. In the case
  of $\Delta=3$, the best previous result is an FPRAS for counting
  proper 5-colorings. We note that there is a barrier to go beyond
  $q = \Delta + 2$ for single-site Glauber dynamics based FPRAS and we
  overcome this by correlation decay approach. Moreover, we develop a
  number of new techniques for the correlation decay approach which
  can find applications in other approximate counting problems.
\end{abstract}

\let\svthefootnote\thefootnote \let\thefootnote\relax\footnote{This
  work was done in part while some of the authors were visiting the
  Simons Institute for the Theory of Computing.}
\addtocounter{footnote}{-1}\let\thefootnote\svthefootnote




\section{Introduction}
The problem of counting proper $q$-colorings has been extensively
studied in computer science and statistical physics. It is known to be
$\sharpP$-hard for $q\ge 3$ even on graphs with bounded maximum degree
$\Delta\geq 3$~\cite{bubley1999approximately}. A number of literature
has been devoted to the design of approximation
algorithms~\cite{bubley1997path,molloy2004glauber,hayes2003non,hayes2003randomly,jerrum1995very,bubley1999approximately,vigoda2000improved,dyer2003randomly,dyer2013randomly,hayes2006coupling,dyer2006randomly}. The
main algorithmic tool used in these works is the method of Markov
chain Monte Carlo (MCMC), which is based on the simulation of a Markov
chain on all proper $q$-colorings of a graph $G$ whose stationary
distribution is the uniform distribution. Although the Markov chains
themselves are usually quite simple, it is challenging to prove the
rapid mixing property of the chains and the interplay between the
number $q$ of colors and the maximum degree $\Delta$ of the graph $G$
turns out to be a key measure for such property to hold.

The \emph{Glauber dynamics} is a natural Markov chain to sample
colorings and it converges to the uniform distribution as long as
$q\ge \Delta+2$.  Jerrum~\cite{jerrum1995very} and Salas and
Sokal~\cite{salas1997absence} independently showed that the Glauber
dynamics mixes rapidly if $q>2\Delta$. The bound of $2\Delta$ was
considered as a barrier for the analysis of the Glauber dynamics and
was even conjectured as a threshold for the rapid mixing property to
hold for a period of time. Later, the conjecture was refuted by Bubley
et al.~\cite{bubley1999approximately} by showing that the Glauber
dynamics indeed rapidly mixes when $\Delta=3$ and $q=5$. It is worth
to note that this result attains the ergodicity threshold for Glauber
dynamics ($q\ge \Delta+2$) and thus it is the best one can achieve via
this method.  For general $\Delta$, the state-of-the-art requires that
$q> \frac{11}{6}\Delta$~\cite{vigoda2000improved}.

All the above algorithms based on MCMC provide randomized
algorithms. Can we get deterministic approximation algorithms?  A
deterministic FPTAS was obtained in~\cite{gamarnik2012correlation}
when $q \ge 2.8432 \Delta+ \beta$ for some sufficiently large $\beta$
on triangle-free graphs. The bound was improved to
$q\ge 2.581\Delta+1$ on general graphs~\cite{lu2013improved}. These
new deterministic FPTASes are based on the \emph{correlation decay}
techniques.

Correlation decay approach is a relatively new approach to design
approximate counting algorithm comparing to the MCMC method. One
advantage of correlation decay approach is that the resulting
algorithms are deterministic. Moreover, there are quite a few
problems, for which an FPTAS based on correlation decay approach was
provided while no MCMC based FPRAS is known. Among which, the most
successful example is the problem of computing the partition function
of anti-ferromagnetic two-spin
systems~\cite{li2012approximate,sinclair2014approximation,li2013correlation},
including counting independent sets~\cite{weitz2006counting}. The
correlation decay based FPTAS is beyond the best known MCMC based
FPRAS and achieves the boundary of
approximability~\cite{sly2012computational,galanis2012inapproximability},
which is the uniqueness condition of the system. It is an important
and challenging open question to extend this result to
anti-ferromagnetic multi-spin systems. Coloring problem (or
anti-ferromagnetic Potts model at zero temperature in the statistical
physics terminology) is the most important and canonical example for
anti-ferromagnetic multi-spin systems. It was proved that the
uniqueness bound for this system on infinite regular trees is exactly
$q = \Delta + 1$~\cite{jonasson2002uniqueness}. This fact supports the
conjecture that $q = \Delta + 1$ is the optimal bound for approximate
counting in general graphs.

\subsection{Our Results}
Our main result is to introduce new techniques to the correlation
decay based algorithm and provide an FPTAS all the way up to the
optimal bound of $q=\Delta+1$ in the case of $\Delta=3$.
\begin{theorem}\label{thm:main}
  There exists an \FPTAS to compute the number of proper four
  colorings on graphs with maximum degree three.
\end{theorem}
As the first algorithm achieving the optimal bound, we view it as a
substantial step towards the optimal counting algorithms for general
graphs. The contribution is three folds
\begin{itemize}
\item It overcomes an intrinsic barrier of MCMC (Glauber dynamics)
  based algorithms. For the case of $q=\Delta+1$, the Glauber dynamics
  Markov chain is not ergodic and thus its stationary distribution is
  not unique. Nevertheless, we obtained FPTAS based on correlation
  decay technique.
\item We provide a number of new design and analysis technique for
  correlation decay based algorithms, which can be used for general
  graph colorings or even other approximate counting problems.
\item Our analysis is simpler than previous analysis of MCMC
  algorithms in similar settings. Even when the maximum degree $\Delta=3$,
  it is already a very challenging problem to analyze the MCMC
  algorithms. In order to improve from $q=6$ to $q=5$,
  \cite{bubley1999approximately} did a very detailed case by case
  analysis and even require computer to verify the proof. We obtain
  the optimal bound of $q=4$.
\end{itemize}

\subsection{Our Techniques}
The key step in all the proofs of correlation decay analysis is to
prove that a recursive function is contractive. For most of current
known correlation decay based FPTASes for coloring problem, the
following recursion, introduced in~\cite{gamarnik2012correlation}, is
used
\[
  \Pr[G,L]{c(v)=i}
  =\frac{\prod_{k=1}^{d}\tuple{1-\Pr[G_v,L_{k,i}]{c(v_k)=i}}}{\sum_{j\in
      L(v)}\prod_{k=1}^{d}\tuple{1-\Pr[G_v,L_{k,j}]{c(v_k)=j}}}.
\]

The notation $\Pr[G,L]{c(v)=i}$ denotes the marginal probability of
the vertex $v$ to be colored $i$ in an instance $(G,L)$ where $G$ is a
graph and $L$ is a color list that associates each vertex a set of
\emph{feasible} colors. $\Pr[G_v,L_{k,j}]{c(v_k)=j}$ denotes a similar
marginal probability in a modified instance: $G_v$ is the graph
obtained from $G$ by removing $v$ and $L_{k,j}$ is obtained from $L$
by removing color $j$ from the color list of the vertex $v_{w}$ where
$w<k$ and $v_w$ is the $w$-th neighbor of $v$ in some canonical order.
In this recursion, $\Pr[G,L]{c(v)=i}$ can be computed from $dq$
different variables of $\Pr[G_v,L_{k,j}]{c(v_k)=j}$ with
$k=1,2,\cdots,d$ and $j=1,2,\cdots,q$. In all previous analyses, one
view them as $dq$ free and independent variables and then bound the
contraction in the worst case. For each single variable, one use the
same recursion to expand to a set of $dq$ new free and independent
variables. This yields a computation tree of degree $dq$. However, the
expansion of the underlying graph is of degree $d$ and we usually call
this gap the information loss or inefficiency of the
recursion. However, these $dq$ variables are not completely free and
independent. The key new idea of this work is to make use of the
relations among these variables to reduce redundancy and improve the
efficiency of the recursion.  Here are two key observations:
\begin{itemize}
\item For different colors $i$ and $j$, the recursions for
  $\Pr[G,L]{c(v)=i}$ and $\Pr[G,L]{c(v)=j}$ involve exactly the same
  set of $dq$ different variables.
\item For $k=1$, $L_{k,j}$ is identical for different color $j$.
\end{itemize}

Using these two observations, we can further expand the $q$ different
variables $\Pr[G_v,L_{1,j}]{c(v_1)=j}$ with $j=1,2,\cdots,q$ into a
set of $d q$ different variables simultaneously. The expansion here is
$d$ (from $q$ variables to $d q$ different variables) rather than
$dq$. In previous analyses, each single variable of these $q$
different variables will further expand to $dq$ free and independent
variables. The total number becomes $dq^2$.

This can be viewed as a partial two-layer recursion: for a subset of
variables in the one layer recursion, we use the same recursive
function to further expand them. We note that the similar information
loss or inefficiency for recursion appears in many correlation decay
based approximation counting algorithms, and it is the main cause of
the sub-optimality of the analysis. The approach introduced here can
also be applied to improve their analyses and the key is to observe
some relations among the redundant variables and make use of it. In
the current case, the improvement becomes substantial when the number
of variables is small.

Another crucial idea in our proof is to get better bounds for
variables $\Pr[G_v,L_{k,j}]{c(v_k)=j}$ and then we only need to prove
contraction in these bounded range. This idea was used in previous
analyses for counting colorings and many other problems.  However, in
our setting of $q=4$ and $\Delta=3$, these values could be as large
as $1$ and as small as $0$. These are trivial bounds for a probability
in general. Here, we use two observation to refine the bounds.

First, we notice that bound of $1$ can only be achieved at the root of
the recursion tree and for all other variable the value is between $0$
and $\frac{1}{2}$.  The boundaries of $0$ and $\frac{1}{2}$ are both
achievable and thus cannot be improved in general. To overcome this,
we use the following \emph{alternative argument}: When the two bounds
of $0$ and $\frac{1}{2}$ are achieved, we can easily detect it and
thus compute the accurate values without error; otherwise, we can get
better bounds. In the later case, we view the variables achieving $0$
and $\frac{1}{2}$ as parameters rather than variables of the recursion
function as we are sure that there is no errors for them, and just
prove that the degenerated recursion function is contractive with
respect to remaining variables. This is plausible since we have
better bounds for remaining variables.

Last but not the least, as in most of the correlation decay approach,
we use a potential function to amortize the decay rate. It remains
the most important and magic ingredient of the proof. There is no
general method to design potential function. Based on some numerical
computation, we propose a new potential function in the
paper. Comparing to the previous potential functions for coloring
problem, the main new feature of our new function is its
non-monotonicity, which captures the property of the problem. We
remark that a potential function with a similar shape can be used for
general graph coloring problem for similar set of recursions.



\section{Preliminaries and the (New) Recursion}\label{sec:prelim}
\paragraph{List coloring and Gibbs measure}
Although we start with a standard graph coloring instance, where each vertex can choose from the same set of 4 different colors,
we need to modify the color list during our algorithms to get a
list-coloring instance. Therefore we work on list-coloring problem in general.
A list-coloring instance is specified by a graph-list pair $(G,L)$,
where $G=(V,E)$ is an undirected graph and $L:V \rightarrow 2^{[q]}$ associates each vertex $v$ with a color list
$L(v) \subseteq [q]$. A proper coloring of $(G,L)$ is an assignment
$c:V\to[q]$ such that (1) $c(v)\in L(v)$ for every $v\in V$ and (2) no
two ends of an edge share the same color, i.e., $c(u)\ne c(v)$ for
every $e=(u,v)\in E$. 

The \emph{Gibbs measure} is the uniform distribution over all proper
colorings of $(G,L)$. For every vertex $v\in V$ and color $i\in [q]$,
we use $\Pr[G,L]{c(v)=i}$ to denote the marginal probability that the
vertex $v$ is colored $i$ in the Gibbs measure.

In the following, we use $\Delta$ to denote the maximum degree of the
graph. If there exists an efficient algorithm to estimate the marginal
probability $\Pr[G,L]{c(v)=i}$, then one can construct an \FPTAS to
count the number of proper colorings.

\begin{lemma}\label{lem:margtopartition}
  Suppose there exists an algorithm to compute a $(1\pm \eps)$ approximation of  $\Pr[G,L]{c(v)=i}$
  for every list-coloring
  instance $(G,L)$ with $G=(V,E)$, $q=4$, $\Delta=3$, $|L(v)|\geq d_v+1$ for every $v\in V$, and every $i\in [q]$ in time $\mathrm{poly}(\abs{V},\frac{1}{\eps})$.
Then there exists an \FPTAS to compute the number of proper
  $4$-colorings on graphs with maximum degree three.
\end{lemma}

The proof Lemma \ref{lem:margtopartition} is routine, see
e.g. \cite{gamarnik2012correlation}. Therefore, the remaining task is to approximate $\Pr[G,L]{c(v)=i}$  for instances
satisfying the conditions stated in Lemma~\ref{lem:margtopartition}.

\paragraph{Recursion}

Let $(G,L)$ be an instance of list-coloring and $v\in V$ be a
vertex. Let $N(v)=\set{v_1,\dots,v_d}$ denote the set of neighbors of
$v$ in $G$, where $d$ is the degree of $v$ and let $G_v$ be the graph
obtained from $G$ by removing vertex $v$ and all its incident
edges. For every $k\in[d]$ and $i\in[q]$, let
\begin{equation}
  \label{eq:list}
  L_{k,i}(u)=
  \begin{cases}
    L(u)\setminus\set{i}, & \mbox{if $u=v_\ell$ for some $\ell<k$}, \\
    L(u), & \mbox{otherwise}
  \end{cases}
\end{equation}
be color lists.  Then the following recursion for computing
$\Pr[G,L]{c(v)=i}$ first appeared in \cite{gamarnik2012correlation}.
\begin{lemma}
  Assuming above notations we have
  \begin{equation}
    \label{eq:recursion}
    \Pr[G,L]{c(v)=i}=\frac{\prod_{k=1}^{d}\tuple{1-\Pr[G_v,L_{k,i}]{c(v_k)=i}}}{\sum_{j\in
        L(v)}\prod_{k=1}^{d}\tuple{1-\Pr[G_v,L_{k,j}]{c(v_k)=j}}}
  \end{equation}
\end{lemma}
\begin{proof}
  Let $Z_{G,L}$ denote the number of proper colorings on $(G,L)$ and
  let $Z_{G,L}(c(u)=i)$ denote the number of proper colorings on
  $(G,L)$ that assigns vertex $u$ with color $i$ for every vertex
  $u\in V$ and every color $i\in [q]$. Then
  \begin{align*}
    \Pr[G,L]{c(v)=i}
    &=\frac{Z_{G,L}(c(v)=i)}{\sum_{j\in L(v)}Z_{G,L}(c(v)=j)}
      =\frac{Z_{G_v,L}\tuple{\bigwedge_{k\in[d]}c(v_k)\ne i}}{\sum_{j\in L(v)}Z_{G_v,L}\tuple{\bigwedge_{k\in[d]}c(v_k)\ne j}}\\
    &=\frac{\Pr[G_v,L]{\bigwedge_{k\in[d]}c(v_k)\ne i}}{\sum_{j\in L(v)}\Pr[G_v,L]{\bigwedge_{k\in [d]}c(v_k)\ne j}}
      =\frac{\prod_{k=1}^d\Pr[G_v,L]{c(v_k)\ne i\mid \bigwedge_{\ell<k}c(v_\ell)\ne i}}{\sum_{j\in L(v)}\prod_{k=1}^d\Pr[G_v,L]{c(v_k)\ne j\mid \bigwedge_{\ell<k}c(v_\ell)\ne j}}\\
    &=\frac{\prod_{k=1}^{d}\tuple{1-\Pr[G_v,L_{k,i}]{c(v_k)=i}}}{\sum_{j\in
      L(v)}\prod_{k=1}^{d}\tuple{1-\Pr[G_v,L_{k,j}]{c(v_k)=j}}}.
  \end{align*}
\end{proof}

Then we can apply the same recursion to further expand
$\Pr[G_v,L_{k,j}]{c(v_k)=j}$ so on and so forth. It gives a
computation tree to compute the value of the root
$\Pr[G,L]{c(v)=i}$. The condition that $q=4$, $\Delta=3$, and
$|L(v)|\geq d_v+1$ for every $v\in V$, holds for all the list-coloring
instances appearing in this computation tree. In the definition of new
color lists (\ref{eq:list}), the list size is decreased by one only
for the neighbors of $v$, but the degrees of its neighbors are also
decreased by one in the new modified graph $G_v$ since we have removed
vertex $v$ and all its incident edges. Therefore the condition
$|L(v)|\geq d_v+1$ remains satisfied for every $v\in V$ in the new
instance. For every probability $\Pr[G',L']{c(u)=j}$ in the
computation tree except the root, the degree $d_u\leq \Delta-1=2$
since we come to this instance by removing a neighbor of $u$ and thus
the degree is decreased by at least one. All these observations are
used in previous analyses.  A more subtle and crucial new observation
is that for every probability $\Pr[G',L']{c(u)=j}$ in the computation
tree except the root, one have $|L(u)|\geq d_u+2$ (which is stronger
than $|L(u)|\geq d_u+1$ ) since the degree of $u$ is decreased by one
while color list for $u$ remains in the definition of (\ref{eq:list}).

We do not analyze this computation tree directly but turn to a more
efficient one by taking the relation between variables into
account. In the definition (\ref{eq:list}) of $L_{k,i}$ , if $k=1$ the
new color lists remain the same for all the remaining vertexes and
thus is independent from the color $i$. Therefore, the $|L(v_1)|$
variables $\Pr[G_v,L_{1,j}]{c(v_1)=j}$ are simply the marginal
probabilities of vertex $v_1$ for different colors in the same
instance. Therefore, when we further expand these variables, they
involve same set of variables. We make use of this property and
further expand these variables as follows.  Let $d_1$ be the degree of
$v_1$ in the graph $G_v$ and $u_1,u_2,\dots,u_{d_1}$ be the neighbors
of $v_k$ in the graph $G_{v}$.  We use $G_{v,v_1}$ to denote the graph
obtained from $G_v$ by removing the vertex $v_1$ and all its incident
edges. For every $k\in[d_1]$ and $i\in[q]$, we use $L'_{k,i}$ to
denote the color list such that
\[
  L'_{k,i}(u)=
  \begin{cases}
    L(u)\setminus\set{i}, & \mbox{if $u=u_\ell$ for some $\ell<k$},\\
    L(u), & \mbox{otherwise}.
  \end{cases}
\]
Applying recursion \eqref{eq:recursion}, we obtain for every
$j \in L(v_1)$, it holds that
\begin{equation}
  \label{eq:recursion2}
  \Pr[G_v,L_{1,j}]{c(v_1)=j}=\frac{\prod_{k=1}^{d_1}\tp{1-\Pr[G_{v,v_1},L'_{k,j}]{c(u_k)=j}}}{\sum_{l\in L(v_1)}\prod_{k=1}^{d_1}\tp{1-\Pr[G_{v,v_1},L'_{k,l}]{c(u_k)=l}}}.
\end{equation}

Then we substitute these into recursion \eqref{eq:recursion} and get a
new recursion for $\Pr[G,L]{c(v)=i}$. We view this new recursion as
one step in the computation tree and analyze its correlation decay
property.  From the algorithmic point of view, this does not make much
difference but it do impact the analysis a lot. A similar situation
appeared in~\cite{liu2015fptas}, where one use the same algorithm to
compute the number of independent sets in bipartite graphs as in
general graphs. However, in that analysis, one combined two step of
the recursion, and viewed it as one single step in the computation
tree, and then analyze the contractive rate directly. Here, we analyze
the partial two-step recursion, where one only further expand the
variables for its first neighbor.

\section{Algorithm}
In this section, we describe our algorithm to estimate marginals.

The main idea of our algorithm to estimate $\Pr[G,L]{c(v)=i}$ is to
recursively apply recursions \eqref{eq:recursion} and
\eqref{eq:recursion2} up to some depth $D$. For the convenience of
analysis, we distinguish between cases, depending on the degree of $v$
and its neighbors.

\begin{itemize}
\item Our algorithm terminates in one of the following three boundary
  cases. (1) the color $i$ is not in the color list $L(v)$, i.e.,
  $i\not\in L(v)$, in which case we return $0$; (2) the recursion
  depth is zero, in which case we return $\frac{1}{\abs{L(v)}}$ and
  (3) the degree of $v$ in $G$ is zero, i.e. $v$ is an isolated
  vertex, in which case we return $\frac{1}{\abs{L(v)}}$.
\item If the degree of $v$ in $G$ is one, the algorithm branches into
  three cases according to the size of $L(v)$. In the case of
  $\abs{L(v)}=2$, we directly apply recursion \eqref{eq:recursion}. In
  the case of $\abs{L(v)}=4$, note that the sum of the marginal
  probabilities of colors $j\in L(v)$ on $v_1$ in $G_v$ is $1$, the
  denominator of the recursion \eqref{eq:recursion} becomes a constant
  $3$. For the same reason, in the case of $\abs{L(v)}=3$, we can
  denote the denominator of the recursion \eqref{eq:recursion} by
  $2+y$, where $y$ is the marginal probability of color
  $j\in[4]\setminus L(v)$ (the absent color) on $v_1$ in $G_v$.
\item If the degree of $v$ in $G$ is two or three, we faithfully apply
  recursion \eqref{eq:recursion} and \eqref{eq:recursion2} to estimate
  the marginals. In order to simplify the analysis, we use the
  following convention in the case of $\deg[G]{v}=2$: Let the
  neighbors of $v$ be $v_1,v_2$, then we always assume
  $\deg[G]{v_1}\ge\deg[G]{v_2}$ and if $\deg[G]{v_1}=\deg[G]{v_2}=1$,
  then $i\not\in L(v_1)$ implies $i\not\in L(v_2)$.
\end{itemize}

The whole algorithm is described below. We use procedure $P(G, L, v, i, D)$ to estimate $\Pr[G,L]{c(v)=i}$ up
to depth $D$.

\bigskip
\begin{algorithm}[H]\label{algo:marg-main}
  \caption{Estimate $\mathbf{Pr}_{G,L}{\left[c(v)=i\right]}$}
  \SetKwInOut{Input}{Input}\SetKwInOut{Output}{Output} \Input{Graph
    $G$; color lists $L$; vertex $v$; color $i$; recursion depth $D$;
  }
  \Output{$P \in [0,1]$: Estimate of $\Pr[G,L]{c(v)=i}$ up to depth
    $D$.}%
  \BlankLine
  \textbf{Function} $P(G,L,v,i,D)$\\
  \Begin{
    \If{$i \notin L(v)$}{ \Return{
        $0$}\;}%
    \If{$D \leq0$}{ \Return{ $\frac{1}{|L(v)|}$}\;}%
    \If{$\deg[G]{v}=0$}{ \Return{ $\frac{1}{\abs{L(v)}}$}\;}%
    \If{$\deg[G]{v}=1$}{ \Return{$P1(G,L,v,i,D)$}\;}
    \If{$\deg[G]{v}=2$}{ \Return{$P2(G,L,v,i,D)$}\;}
    \If{$\deg[G]{v}=3$}{\Return{$P3(G,L,v,i,D)$}\;}
  }

\end{algorithm}
The procedures $P1(G,L,v,i,D)$, $P2(G,L,v,i,D)$ and $P3(G,L,v,i,D)$
deal with the case of $\deg[G]{v}=1$, $\deg[G]{v}=2$ and
$\deg[G]{v}=3$ respectively.

\paragraph{Case $\deg[G]{v}=1$:}~

\begin{algorithm}[H]\label{algo:marg-deg1}
  \caption{Estimate $\mathbf{Pr}_{G,L}{\left[c(v)=i\right]}$ when $\mathrm{deg}_G(v)=1$}
  \textbf{Function} $P1(G, L, v, i, D)$\\
  \Begin{ \tcc{the vertex $v$ has only one neighbor
      $v_1$.}  \BlankLine \tcc{If the $L(v)=\set{i,j}$.}
    \If{$\abs{L(v)}=2$}{ $x\gets P(G_v,L_{1,i},v_1,i,D-1)$\;
      $y\gets P(G_v,L_{1,j},v_1,j,D-1)$\;
      \Return{$\frac{1-x}{2-x-y}$}\; } \If{$\abs{L(v)}=4$}{
      $x\gets P(G_v,L_{1,i},v_1,i,D-1)$\; \Return{$\frac{1-x}{3}$}\; }
    \tcc{in the following case, $\abs{L(v)}=3$.}
    \If{$i\in L(v_1)$}{ Let $j$ be the color in the singleton set
      $[4]\setminus L(v)$\; $x\gets P(G_v,L_{1,i},v_1,i,D-1)$\;
      $y\gets P(G_v,L_{1,j},v_1,j,D-1)$\; \Return{$\frac{1-x}{2+y}$}\;
    } }
\end{algorithm}

\paragraph{Case $\deg[G]{v}=2$:}~

\begin{algorithm}[H]\label{algo:marg-deg2}
  \textbf{Function} $P2(G, L, v, i, D)$\\
  \Begin{
    \tcc{the vertex $v$ has two neighbors $\set{v_1,v_2}$ with
      $\deg[G]{v_1}\ge\deg[G]{v_2}$; the vertex $v_1$ has neighbors
      $\set{u_1,\dots,u_{d_1}}$ in the graph $G_v$. We also assume
      that if $\deg[G]{v_1}=\deg[G]{v_2}=1$, then $i\not\in L(v_1)$
      implies $i\not\in L(v_2)$.}%
    \For{$j\in L(v)$}{ \If{$j\not\in L(v_1)$}{ $f_j\gets 0$\; } \Else{
        \For{$k\in[d_1]$}{ \For{$w\in L(v_1)$}{
            $x_{k,w}\gets P(G_{v,v_1},L'_{k,w},u_k,w,D-1)$\; } }
        $f_j\gets\frac{\prod_{k=1}^{d_1}\tp{1-x_{k,j}}}{\sum_{w\in
            L(v_1)}\prod_{k=1}^{d_1}\tp{1-x_{k,w}}}$\; }
      $y_j\gets P(G_v,L_{2,j},v_2,j,D-1)$\; }
    \Return{$\frac{(1-f_i)(1-y_i)}{\sum_{j\in L(v)}(1-f_j)(1-y_j)}$\;}
    
  }
  \caption{Estimate $\mathbf{Pr}_{G,L}{\left[c(v)=k\right]}$ when $\mathrm{deg}_{G}\tp{v}=2$}
\end{algorithm}

\paragraph{Case $\deg[G]{v}=3$:}~

\begin{algorithm}[H]\label{algo:marg-deg3}
  \textbf{Function} $P3(G, L, v, i, D)$\\
  \Begin{
    \tcc{the vertex $v$ has three neighbors $\set{v_1,v_2,v_3}$.}%
    \For{$j\in L(v)$}{ $x_j\gets P(G_v,L_{1,j},j,D)$\;
      $y_j\gets P(G_v,L_{2,j},j,D)$\; $z_j\gets P(G_v,L_{3,j},j,D)$\;
    }
    \Return{$\frac{(1-x_i)(1-y_i)(1-z_i)}{\sum_{j\in
          L(v)}(1-x_j)(1-y_j)(1-z_j)}$}\; }
  \caption{Estimate $\mathbf{Pr}_{G,L}{\left[c(v)=i\right]}$ when $\mathrm{deg}_G\tp{v}=3$}
\end{algorithm}

\begin{proposition}\label{prop:sumone}
  Let $q=4$. Given a list-coloring instance $(G=(V,E),L)$ with maximum
  degree 3, a vertex $v \in V$ satisfying $\deg[G]{v} \le 2$ and
  $|L(v)| \ge \deg[G]{v}+2$, a nonnegative integer $D$, we have
  \[ \sum_{i=1}^{4} P(G,L,v,i,D) = 1 \]
\end{proposition}
\begin{proof}
  We will prove by induction on $D$. When $D=0$ we have
  $\sum_{i=1}^{4}P(G,L,v,i,0)=\sum_{i \in
    L(v)}\frac{1}{|L(v)|}=1$. Suppose the proposition holds for
  $D-1$. To obtain the proof for $D$, we will discuss on degree of
  $v$.
  \begin{enumerate}
  \item $\deg[G]{v}=0$.

    Clearly
    $\sum_{i=1}^{4}P(G,L,v,i,D)=\sum_{i \in L(v)}\frac{1}{|L(v)|}=1$.
	
  \item $\deg[G]{v}=1$.

    Let $x_i=P(G_v,L_{1,i},v_1,i,D-1)$. By definition we have
    $L_{1,i}=L$ for all $i \in [4]$. Therefore
    $\sum_{i=1}^{4}x_i=\sum_{i=1}^{4}P(G_v,L,v_1,i,D-1)=1$ by
    induction hypothesis.
	
    If $|L(v)|=4$,
    $\sum_{i=1}^{4}P(G,L,v,i,D)=\sum_{i=1}^{4}\frac{1-x_i}{3}=\frac{4-1}{3}=1$.
	
    If $|L(v)|=3$, assume $j \notin L(v)$. Then
    $\sum_{i=1}^{4}P(G,L,v,i,D)=\sum_{i \in
      L(v)}\frac{1-x_i}{2+x_j}=\frac{3-(1-x_j)}{2+x_j}=1$.
	
  \item $\deg[G]{v}=2$.

    In this case $|L(v)|=4$. So
    $\sum_{i=1}^{4}P(G,L,v,i,D)=\sum_{i=1}^{4}\frac{(1-f_i)(1-y_i)}{\sum_{j
        \in
        L(v)}(1-f_j)(1-y_j)}=\frac{\sum_{i=1}^{4}(1-f_i)(1-y_i)}{\sum_{j=1}^{4}(1-f_j)(1-y_j)}=1$.
  \end{enumerate}
\end{proof}

Using the same proof, we can also have $\sum_{j=1}^{4}f_j=1$, where
$f_j$ is defined in Algorithm \ref{algo:marg-deg2}.

\bigskip
We conclude this section with the following lemma, whose proof is postponed to Section~\ref{sec:proofmain}.
\begin{lemma}\label{lem:approx}
  Let $q=4$. There exists an algorithm such that for every
  list-coloring instance $(G,L)$ with $G=(V,E)$ and maximum degree at
  most three, every vertex $v\in V$, every coloring $i\in L(v)$ and
  every $0<\eps<1$, it computes a number $\hat p$ in time
  $\mathrm{poly}\!\tp{\abs{V},\frac{1}{\eps}}$ satisfying
  \[
    (1-\eps)\hat p\le\Pr[G,L]{c(v)=i}\le (1+\eps)\hat p.
  \]
\end{lemma}



\section{Bounds}
In this section, we introduce upper and lower bounds for values
computed in the algorithm. These bounds will play a crucial role in our
proof.

\begin{definition}
  We call a a triple $(G=(V,E),L, v\in V)$ (a list-coloring instance
  together with a vertex in the graph) \textbf{reachable} if the following
  condition is satisfied: $\deg[G]{u}\le 3$ and
  $|L(u)|\geq \deg[G]{u}+1$ for every $u\in V$, $\deg[G]{v}\le 2$ and
  $\abs{L(v)}\ge \deg[G]{v}+2$.
\end{definition}

It follows from the discussion in Section \ref{sec:prelim} that for all the
probability $\Pr[G,L]{c(v)=i}$ appeared in the computation tree except
the root, $(G,L,v)$ is reachable.

\begin{proposition}\label{prop:bd0}
  Let $(G,L,v)$ be reachable, $i\in[4]$ be a color , and $D$ be a
  nonnegative integer. Then it holds that
  \[ 0 \le P(G,L,v,i,D) \le \frac{1}{2}. \]
\end{proposition}
\begin{proof}
  We prove by induction on $D$. For base case, $P(G,L,v,i,D)$ will
  return $\frac{1}{|L(v)|}$ if $D=0$, so the proposition holds since
  $|L(v)| \ge 2$.
	
  Suppose the proposition holds for $D-1$. We discuss on degree of
  $v$.
  \begin{enumerate}
  \item $\deg[G]{v}=0$.
		
    In this case $P(G,L,v,i,D)$ will return $\frac{1}{|L(v)|}$ where
    $|L(v)| \ge \deg[G]{v}+2 = 2$, hence we have
    $0 \le P(G,L,v,i,D) \le \frac{1}{2}$.
		
  \item $\deg[G]{v}=1$.
		
    Let $x=P(G_v,L_{1,i},v_1,i,D-1)$ and $y=P(G_v,L_{1,j},j,D-1)$, as
    defined in Algorithm~\ref{algo:marg-deg1}. Then
    $0 \le x,y \le \frac{1}{2}$ by induction hypothesis.
		
    According to algorithm, $P(G,L,v,i,D)$ will return $\frac{1-x}{3}$
    or $\frac{1-x}{2+y}$, and in both cases this return value is
    bounded by $\frac{1}{2}$ given $x,y \ge 0$.
		
  \item $\deg[G]{v}=2$.
		
    Let $f_j$, $x_{k,w}$ and $y_j$ be the variables defined in
    Algorithm~\ref{algo:marg-deg2}. By induction hypothesis we have
    $0 \le x_{k,w},y_j \le \frac{1}{2}$. As for $f_j$, we need to
    further discuss on $d_1$.
		
    If $d_1 \in \set{0,1}$ then $f_j \le \frac{1}{2}$ immediately
    follows, as we have already seen in previous two cases. If
    $d_1=2$, we also have
    \begin{align*}
      f_j &= \frac{\prod_{k=1}^{2}(1-x_{k,j})}{\sum_{w \in L(v_1)}\prod_{k=1}^{2}(1-x_{k,w})} \\
          &\le \frac{1-x_{1,j}}{(1-x_{1,j})+\frac{1}{2}\sum_{w \in L(v_1) \setminus \set{j}}(1-x_{1,w})} \\
          &= \frac{1-x_{1,j}}{(1-x_{1,j})+\frac{1}{2}\tuple{2+x_{1,j}}} \\
          &\le \frac{1}{2}.
    \end{align*}
    Here we used the fact that $\sum_{w \in L(v_1)}x_{1,w}=1$, since
    $|L(v_1)|=4$ when $d_1=2$. Similarly we have
    \[ P(G,L,v,i,D) = \frac{(1-f_i)(1-y_i)}{\sum_{j \in
          L(v)}(1-f_j)(1-y_j)} \le \frac{1}{2}. \]
		
  \end{enumerate}
\end{proof}

\begin{proposition}\label{prop:bd1} Let $(G,L,v)$ be reachable,
  $i\in L(v)$ be a color , and $D$ be a nonnegative integer. Then it
  holds that
  \begin{enumerate}[(1)]
  \item if $\deg[G]{v}=2$, then
    \[ P(G,L,v,i,D) \ge \frac{1}{13}; \]
  \item if $\deg[G]{v}\le1$, then
    \[ P(G,L,v,i,D) \ge \frac{1}{6}. \]
  \end{enumerate}
\end{proposition}
\begin{proof}
  If $\deg[G]{v}=0$ or $D=0$, we have
  $P(G,L,v,i,D)=\frac{1}{|L(v)|} \geq \frac{1}{4}$. In the following,
  we assume $D\geq 1$ and $\deg[G]{v}\geq 1$.  We discuss on degree of
  $v$.
  \begin{enumerate}[(1)]
  \item $\deg[G]{v}=2$

    It must be the case that $|L(v)|=4$. Therefore we have
    \begin{align*}
      P(G,L,v,i,D) &= \frac{(1-f_i)(1-y_i)}{(1-f_i)(1-y_i)+\sum_{j \in L(v) \setminus \set{i}}(1-f_j)(1-y_j)}  \\
                   &\ge \frac{\tp{1-\frac{1}{2}}^2}{\tp{1-\frac{1}{2}}^2+\sum_{j \in L(v) \setminus \set{i}}(1-0)} 
                   = \frac{\frac{1}{4}}{\frac{1}{4}+3} = \frac{1}{13}.
    \end{align*}
    The upper bound $\frac{1}{2}$ for $f_j$ and $y_j$ is guaranteed by
    Proposition \ref{prop:bd0}.
      
  \item $\deg[G]{v}=1$

    If $|L(v)|=4$,
    $P(G,L,v,i,D)=\frac{1-x}{3}\ge\frac{1-\frac{1}{2}}{3}=\frac{1}{6}$.
	
    If $|L(v)|=3$,
    $P(G,L,v,i,D)=\frac{1-x}{2+y}\ge\frac{1-\frac{1}{2}}{2+\frac{1}{2}}=\frac{1}{5}>\frac{1}{6}$.

  \end{enumerate}
\end{proof}

Note that overall we have lower bounds $\frac{1}{13}$ for
$P(G,L,v,i,D)$, regardless of the degree of $v$. Furthermore, for
$f_j$'s defined in Algorithm~\ref{algo:marg-deg2} we can draw a
similar conclusion: if $j \in L(v_1)$ then $f_j \ge \frac{1}{13}$.

\begin{proposition}\label{prop:bd2}
  Let $(G,L,v)$ be reachable and $D$ be a nonnegative integer. Then
  for every color $i\in[4]$ such that $i \in L(u)$ for some neighbor
  $u$ of $v$, we have
  \begin{enumerate}[(1)]
  \item if $\deg[G]{v}=2$, then
    \[ P(G,L,v,i,D) \le \frac{12}{25}; \]
  \item if $\deg[G]{v}=1$, then
    \[ P(G,L,v,i,D) \le \frac{6}{13}. \]
  \end{enumerate}
\end{proposition}
\begin{proof}
  If $i \notin L(v)$ then $P(G,L,v,i,D)$ returns 0 and we are done. If
  $D=0$ and $i\in L(v)$, then
  $P(G,L,v,i,0)= \frac{1}{|L(v)|}\leq \frac{1}{3}$ since $v$ have at
  least one neighbor.  In the following, we assume $i \in L(v)$ and
  $D\ge 1$.

  \begin{enumerate}[(1)]
  \item $\deg[G]{v}=2$

    If $i \in L(v_1)$, by Proposition~\ref{prop:bd1} we know that
    $f_i \ge \frac{1}{13}$.
    \begin{align*}
      P(G,L,v,i,D)
      &= \frac{(1-f_i)(1-y_i)}{(1-f_i)(1-y_i)+\sum_{j \in L(v) \setminus \set{i}}(1-f_j)(1-y_j)} \\
      &\le \frac{(1-f_i)(1-0)}{(1-f_i)(1-0)+\sum_{j \in L(v) \setminus \set{i}}(1-f_j)\tp{1-\frac{1}{2}}} \\
      &= \frac{2(1-f_i)}{2(1-f_i)+(3-\sum_{j \in L(v) \setminus \set{i}}f_j)} \\
      &= \frac{2(1-f_i)}{4-f_i} \le \frac{24}{51} < \frac{12}{25}.
    \end{align*}
    Here we used the fact that $\sum_{j \in L(v)}f_j=1$. On the other
    hand, if $i \in L(v_2)$ then $y_i \ge \frac{1}{13}$. So
    \begin{align*}
      P(G,L,v,i,D)
      &= \frac{(1-f_i)(1-y_i)}{(1-f_i)(1-y_i)+\sum_{j \in L(v) \setminus \set{i}}(1-f_j)(1-y_j)} \\
      &\le \frac{(1-f_i)\tp{1-\frac{1}{13}}}{(1-f_i)\tp{1-\frac{1}{13}}+\sum_{j \in L(v) \setminus \set{i}}(1-f_j)\tp{1-\frac{1}{2}}} \\
      &=
        \frac{\frac{12}{13}(1-f_i)}{\frac{12}{13}(1-f_i)+\frac{1}{2}\sum_{j
        \in L(v) \setminus \set{i}}(1-f_j)}
        = \frac{24(1-f_i)}{50-11f_i} \le \frac{12}{25}.
    \end{align*}

  \item $\deg[G]{v}=1$
	
    Clearly $i \in L(u)$ where $u$ is the only neighbor of $v$. So
    $x=P(G_v,L,u,i,D-1)\ge\frac{1}{13}$.
	
    If $|L(v)|=4$ then
    $P(G,L,v,i,D) = \frac{1-x}{3} < \frac{1}{3} < \frac{6}{13}$.
	
    If $|L(v)|=3$ then
    $P(G,L,v,i,D) = \frac{1-x}{2+y} \le \frac{1-\frac{1}{13}}{2} =
    \frac{6}{13}$.

  \end{enumerate}

\end{proof}

\begin{proposition}\label{prop:bd3} Let $(G,L,v)$ be reachable,
  $i\in[4]$ be a color. Assume $\deg[G]{v}=2$, then one of the following
  holds:
  \begin{enumerate}[(1)]
  \item the vertex $v$ and its two neighbors form a triangle in $G$;
  \item $\P(G,L,v,i,D) \le \frac{13}{27}$ for all integer $D\ge 2$.
  \end{enumerate}
\end{proposition}
\begin{proof}
  Without loss of generality we assume $i=1$.  Denote by $v_1$ and
  $v_2$ the two neighbors of $v$ in $G$. We only need to consider the
  case when $1 \notin L(v_1)$ and $1 \notin L(v_2)$, i.e. $f_1=y_1=0$,
  since otherwise by Proposition~\ref{prop:bd2} we immediately have
  $P(G,L,v,1,D) \le \frac{12}{25} < \frac{13}{27}$. In this case $v_1$
  and $v_2$ each only have up to one neighbor in $G_v$, which we will
  denote by $u_1$ and $u_2$ respectively. We now continue to discuss
  in two cases.
  \begin{enumerate}[(1)]
  \item $\deg[G_v]{v_2}=0$

    According to the algorithm $y_j=\frac{1}{|L(v)|}$ for
    $j \in L(v_2)$ and $y_j=0$ for $j \notin L(v_2)$.
	
    If $|L(v_2)|=3$ then $y_1=0$ and $y_2=y_3=y_4=\frac{1}{3}$. We
    have
    \begin{align*}
      P(G,L,v,i,D) &= \frac{1}{1+\sum_{j \in L(v) \setminus \set{1}}(1-f_j)(1-y_j)} \\
                   &= \frac{1}{1+\frac{2}{3}\tp{3-\sum_{j \in L(v) \setminus \set{1}}f_j}} \\
                   &= \frac{1}{1+\frac{4}{3}} 
                   = \frac{3}{7} < \frac{13}{27}.	
    \end{align*}
	
    If $|L(v_2)|=2$ we can assume $2 \notin L(v_2)$, thus $y_1=y_2=0$
    and $y_3=y_4=\frac{1}{2}$. We have
    \begin{align*}
      P(G,L,v,i,D) &= \frac{1}{1+\sum_{j \in L(v) \setminus \set{i}}(1-f_j)(1-y_j)} \\
                   &\le \frac{1}{1+(1-f_2)(1-0)+(1-f_3)(1-\frac{1}{2})+(1-f_4)(1-\frac{1}{2})} \\
                   &= \frac{1}{1+\frac{1}{2}(1-f_2)+\frac{1}{2}\sum_{j \in L(v) \setminus \set{1}}(1-f_j)} \\
                   &= \frac{1}{2+\frac{1}{2}(1-f_2)}
                     \le \frac{4}{9} < \frac{13}{27}.
    \end{align*}
	
  \item $\deg[G_v]{v_2}=1$.

    Since $v$, $v_1$ and $v_2$ do not form a triangle,
    $L_{2,k}(u_2)=L(u_2)$ for all color $k$, thus it follows from
    Proposition~\ref{prop:bd2} that for every
    $j\in L_{2,j}(u_2)=L(u_2)$, we have
    \[
      y_j = P(G_v,L_{2,j},v_2,j,D-1) \le \frac{6}{13}.
    \]
    So if $L(v_2) \subseteq L(u_2)$ we have
    \begin{align*}
      P(G,L,v,1,D)
      &=  \frac{1}{1+\sum_{j \in L(v) \setminus \set{1}}(1-f_j)(1-y_j)} \\
      &\le  \frac{1}{1+\tp{1-\frac{6}{13}}\sum_{j \in L(v) \setminus \set{1}}(1-f_j)} \\
      &=  \frac{1}{1+\frac{7}{13} \cdot 2}
        =  \frac{13}{27}.
    \end{align*}
    On the other hand, consider $L(v_2)\not\subseteq L(u_2)$. Notice
    that $1 \notin L(v_2)$ so there should be some other color, say
    color 2, satisfying $2 \in L(v_2) \setminus L(u_2)$. This also
    forces $\deg[G_{v,v_2}]{u_2} \le 1$. Let
    \[
      t_{kj} \defeq P(G_{v,v_2},L_{2,k},u_2,j,D-2),
    \]
    where $G_{v,v_2}\defeq \tuple{G_v}_{v_2}$, i.e., the graph
    obtained from $G$ by removing $v$ and $v_2$ and all edges incident
    to them. Since $2 \notin L(u_2)$ we have $t_{k2}=0$ for all
    $k$. We need to further distinguish between two cases.
    \begin{enumerate}[(i)]
    \item $1 \in L(u_2)$.

      Recall that $L_{2,k}(u_2)=L(u_2)$ so $\forall k \in L(v)$,
      $1 \in L_{2,k}(u_2)$. Combining $\deg[G_{v,v_2}]{u_2} \le 1$, by
      Proposition~\ref{prop:bd1} we have $\forall k \in L(v)$,
      $t_{k1} \ge \frac{1}{6}$. Specifically we have
      $t_{21} \ge \frac{1}{6}$. Now $1$ is the color in singleton set
      $[4] \setminus L(v_2)$, so according to Algorithm~\ref{algo:marg-deg1}
      \[ y_2 = \frac{1-t_{22}}{2+t_{21}} \le \frac{6}{13}. \] As a
      consequence, we again have $y_j \le \frac{6}{13}$ for all
      $j \in L(v_2)$ and the theorem follows.

    \item $1 \notin L(u_2)$.  In this case $u_2$ is isolated in
      $G_{v,v_2}$ with color list $\set{3,4}$.  So it is clear that
      $t_{k1}=t_{k2}=0$ and $t_{k3}=t_{k4}=\frac{1}{2}$ for every $k$.
      Further we have $y_2=\frac{1}{2}$ and $y_3=y_4=\frac{1}{4}$. Now
      \begin{align*}
        P(G,L,v,1,D)
        &= \frac{1}{1+\sum_{j \in L(v) \setminus \set{1}}(1-f_j)(1-y_j)} \\
        &= \frac{1}{1+(1-\frac{1}{2})(1-f_2)+(1-\frac{1}{4})(1-f_3)+(1-\frac{1}{4})(1-f_4)} \\
        &= \frac{1}{1+\frac{3}{4}\sum_{j \in L(v) \setminus \set{1}}(1-f_j)-\frac{1}{4}(1-f_2)} \\
        &= \frac{4}{9+f_2}
          \le \frac{4}{9} < \frac{13}{27}.
      \end{align*}
    \end{enumerate}
  \end{enumerate}
\end{proof}

Combining above propositions with the $\deg[G]{v}\le 1$ case, we have
the following theorem for bounds on marginal probabilities computed:

\begin{theorem}\label{thm:bounds}
  Let $(G,L,v)$ be reachable, $i\in[4]$ be a color. Then one of the
  following propositions holds:
  \begin{enumerate}[(1)]
  \item $P(G,L,v,i,D) = \frac{1}{2}$ for all integer $D\ge 2$;
  \item $P(G,L,v,i,D) \le \frac{13}{27}$ for all integer $D\ge
    2$. Specifically $P(G,L,v,i,D) \le \frac{6}{13}$ when
    $\deg[G]{v}\le 1$.
  \end{enumerate}
  Furthermore, when $P(G,L,v,i,D) = \frac{1}{2}$ for some integer
  $D\ge 2$ the local structure of $G$ around $v$ falls into one of the
  following three cases (see Figure \ref{fig:bound1},\ref{fig:bound2}
  and \ref{fig:bound3}):
  \begin{enumerate}[(1)]
  \item $\deg[G]{v}=0$.  Then $j,w \notin L(v)$ for two distinct
    colors $j, w$ other than $i$ (Figure~\ref{fig:bound1}).
  \item $\deg[G]{v}=1$.  Denote by $u$ neighbor of $v$.  Then
    $i \notin L(u)$ and $j \notin L(u) \cup L(v)$ for some color
    $j \neq i$ (Figure~\ref{fig:bound2}).
  \item $\deg[G]{v}=2$.  Denote by $u_1$ and $u_2$ two neighbors of
    $v$.  Then $v, u_1, u_2$ form a triangle, and
    $i \notin L(u_1) \cup L(u_2)$ (Figure~\ref{fig:bound3}).
  \end{enumerate}
\end{theorem}
\begin{proof}
  Assume w.l.o.g. $i=1$, and we will assume $1 \in L(v)$, otherwise
  the statement is trivial. From Proposition~\ref{prop:bd3} we know if
  $P(G,L,v,i,D)=\frac{1}{2}$ then $v$ and its two neighbors $v_1$,
  $v_2$ must form a triangle, and $1 \notin L(v_1) \cup L(v_2)$, as
  depicted in Figure~\ref{fig:bound3}. If this is not the case then
  $P(G,L,v,i,D) \le \frac{13}{27}$. Now we focus on those
  $\deg[G]{v} \le 1$ cases.

  \begin{enumerate}[(1)]
  \item $\deg[G]{v}=0$.

    We know $\abs{L(v)} \ge 2$.  If $\abs{L(v)} \ge 3$ then apparently
    $P(G,L,v,1,D) \le \frac{1}{3} < \frac{13}{27}$.  If $|L(v)|=2$
    then $P(G,L,v,i,D)=\frac{1}{2}$ and this is just the case depicted
    in Figure~\ref{fig:bound1}

  \item $\deg[G]{v}=1$.

    By Algorithm \ref{algo:marg-deg1}, if $|L(v)|=4$ then
    $P(G,L,v,1,D)$ will return $\frac{1-x}{3} < \frac{13}{27}$. So we
    focus on $|L(v)|=3$. Assume $2$ is the color in singleton set
    $[4] \setminus L(v)$.

    According to the algorithm, $P(G,L,v,1,D)$ now returns
    $\frac{1-x}{2+y}$ where
    \begin{align*}
      x &= P(G_v,L_{1,1},u,1,D-1) \\
      y &= P(G_v,L_{1,2},u,2,D-1)
    \end{align*}
    Notice that $\frac{1-x}{2+y}$ could reach $\frac{1}{2}$ if and
    only if $x=y=0$. Otherwise at least one of $x$ and $y$ is bounded
    by $\frac{1}{6}$ from below, thus $\frac{1-x}{2+y}$ is bounded by
    $\max\set{\frac{1-0}{2+1/6},\frac{1-1/6}{2+0}}=\frac{6}{13}$.
	  
    Moreover, $x=y=0$ indicates that $1 \notin L_{1,1}(u)$ and
    $2 \notin L_{1,2}(u)$. Recall $L_{1,1} = L_{1,2} = L$, it
    immediately follows that $\deg[G_v]{u}=0$ and $1 \notin L(u)$ and
    $2 \notin L(u)$.  Together with $2 \notin L(v)$ we have
    $2 \notin L(u) \cup L(v)$ which completes the proof (This case is
    depicted in Figure \ref{fig:bound2}).
  \end{enumerate}

\end{proof}



\begin{figure}[H]
  \centering {\begin{tikzpicture}
\usetikzlibrary{positioning}
\tikzstyle{onode}=[draw,circle,scale=1]{};
\tikzstyle{cross}=[path
    picture={\draw[black](path picture bounding box.south east) --
      (path picture bounding box.north west) (path picture bounding
      box.south west) -- (path picture bounding box.north east);}];    
\node [cross,onode] (v1) at (-1,2) {};
\node [onode,fill=red] (v2) at (-1.5,1) {};
\node [onode,fill=green] (v3) at (-0.5,1) {};

\draw  (v1) edge (v2);
\draw  (v1) edge (v3);
\node [font=\footnotesize] at (1.18,2.4) {
	$\Pr{\tikz[baseline=-1.5pt,scale=0.2]{\node [draw, circle,inner sep=1.5pt,cross] (v1) at
        (0,0) {}; }= \tikz[baseline=-1.5pt,scale=0.2] {\node
        [circle,fill=blue,inner sep=1.5pt] at (0,0) {};}}=\frac{1}{2};
    $
};
\node [font=\footnotesize] at (1.2,2) {$\Pr{\tikz[baseline=-1.5pt,scale=0.2]{\node [cross,draw, circle,inner sep=1.5pt] (v1) at
        (0,0) {}; }= \tikz[baseline=-1.5pt,scale=0.2] { \node
        [circle,fill=yellow,inner sep=1.5pt] at (0,0) {}; }}
        =\frac{1}{2}.$};
\end{tikzpicture}
}
  \caption{Boundary case one}
  \label{fig:bound1}
\end{figure}
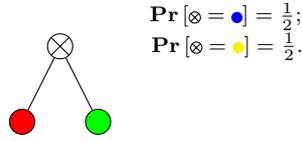
\begin{figure}[H]
  \centering {\begin{tikzpicture}
\usetikzlibrary{positioning}
\tikzstyle{onode}=[draw,circle,scale=1]{};
\tikzstyle{cross}=[path
    picture={\draw[black](path picture bounding box.south east) --
      (path picture bounding box.north west) (path picture bounding
      box.south west) -- (path picture bounding box.north east);}];    
\node [onode] (v1) at (-0.5,2.5) {};
\node [onode,fill=red] (v2) at (-1,1.5) {};
\node [onode,fill=green] (v3) at (0,1.5) {};

\draw  (v1) edge (v2);
\draw  (v1) edge (v3);
\node [font=\footnotesize] at (1.5,4) {
	$\Pr{\tikz[baseline=-1.5pt,scale=0.2]{\node [draw, circle,inner sep=1.5pt,cross] at
        (0,0) {}; }= \tikz[baseline=-1.5pt,scale=0.2] {\node
        [circle,fill=green,inner sep=1.5pt] at (0,0) {};}} 
    =\frac{1}{2};$
};
\node [font=\footnotesize] at (2.35,3.5) {$\Pr{\tikz[baseline=-1.5pt,scale=0.2]{\node [cross,draw, circle,inner sep=1.5pt]  at
        (0,0) {}; }= \tikz[baseline=-1.5pt,scale=0.2] { \node
        [circle,fill=blue,inner sep=1.5pt] at (0,0) {}; }}=\Pr{\tikz[baseline=-1.5pt,scale=0.2]{\node [cross,draw, circle,inner sep=1.5pt]  at
        (0,0) {}; }= \tikz[baseline=-1.5pt,scale=0.2] { \node
        [circle,fill=yellow,inner sep=1.5pt] at (0,0) {}; }}
        =\frac{1}{4}.   $};
\node [cross,onode] (v4) at (-1,3.5) {};
\draw  (v1) edge (v4);
\node [onode,fill=red] (v5) at (-1.5,2.5) {};
\draw  (v4) edge (v5);
\end{tikzpicture}
}
  \caption{Boundary case two}
  \label{fig:bound2}
\end{figure}
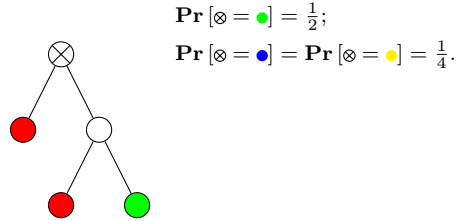
\begin{figure}[H]
  \centering {\begin{tikzpicture}
\usetikzlibrary{positioning}
\tikzstyle{onode}=[draw,circle,scale=1]{};
\tikzstyle{cross}=[path
    picture={\draw[black](path picture bounding box.south east) --
      (path picture bounding box.north west) (path picture bounding
      box.south west) -- (path picture bounding box.north east);}];    
\node [onode] (v1) at (-2,2) {};
\node [onode,fill=red] (v2) at (-2.5,1) {};
\node [onode,fill=red] (v3) at (-0.5,1) {};

\node [font=\small] at (0.5,3.5) {
	$\Pr{\tikz[baseline=-1.5pt,scale=0.2]{\node [draw, circle,inner sep=1.5pt,cross] at
        (0,0) {}; }= \tikz[baseline=-1.5pt,scale=0.2] {\node
        [circle,fill=red,inner sep=1.5pt] at (0,0) {};}}
    =\frac{1}{2};$
};
\node [font=\small] at (1.41,3.1) {$\Pr{\tikz[baseline=-1.5pt,scale=0.2]{\node [cross,draw, circle,inner sep=1.5pt]  at
        (0,0) {}; }= \tikz[baseline=-1.5pt,scale=0.2] { \node
        [circle,fill=green,inner sep=1.5pt] at (0,0) {}; }}=\Pr{\tikz[baseline=-1.5pt,scale=0.2]{\node [cross,draw, circle,inner sep=1.5pt]  at
        (0,0) {}; }= \tikz[baseline=-1.5pt,scale=0.2] { \node
        [circle,fill=blue,inner sep=1.5pt] at (0,0) {}; }}=\frac{1}{6};
           $};
        
 \node [font=\small] at (0.5,2.7) {$\Pr{\tikz[baseline=-1.5pt,scale=0.2]{\node [cross,draw, circle,inner sep=1.5pt]  at
        (0,0) {}; }= \tikz[baseline=-1.5pt,scale=0.2] { \node
        [circle,fill=yellow,inner sep=1.5pt] at (0,0) {}; }}=\frac{1}{6}.   $};       
\node [cross,onode] (v4) at (-1.5,3) {};

\node [onode] (v5) at (-1,2) {};
\draw  (v4) edge (v1);
\draw  (v4) edge (v5);
\draw  (v1) edge (v5);
\draw  (v1) edge (v2);
\draw  (v5) edge (v3);
\end{tikzpicture}
}
  \caption{Boundary case three}
  \label{fig:bound3}
\end{figure}
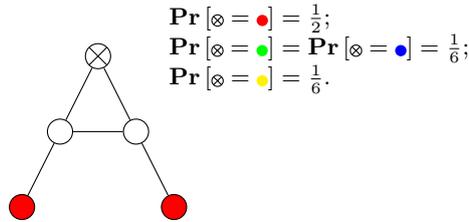

Consider a depth $D$ that is large enough(larger than the size of
$G$), then clearly $P(G,L,v,i,D)$ should return
$\Pr[G,L]{c(v)=i}$. Therefore we can actually draw the same
conclusions for true value $\Pr[G,L]{c(v)=i}$. To make things clearer,
we present the following theorem.

\begin{theorem} \label{thm:consistency} Let $(G,L,v)$ be reachable,
  $i\in[4]$ be a color and $D\geq 2$ be an integer. Then
  \begin{enumerate}
  \item $P(G,L,v,i,D)=0$ if and only if $\Pr[G,L]{c(v)=i}=0$;
  \item $P(G,L,v,i,D)=\frac{1}{2}$ if and only if
    $\Pr[G,L]{c(v)=i}=\frac{1}{2}$;
  \item $P(G,L,v,i,D) \in \left[\frac{1}{13},\frac{13}{27}\right]$ if
    and only if
    $\Pr[G,L]{c(v)=i} \in \left[\frac{1}{13},\frac{13}{27}\right]$.
  \end{enumerate}
\end{theorem}


\section{Correlation Decay}
In this section, we discuss the correlation decay property of our
recursion. First we present the main theorem.

\begin{theorem}\label{thm:cd} Suppose $D \ge 3$ and $q=4$. Let
  $\lambda =\frac{9996}{10000}$ be a constant, then for any list-coloring instance $(G=(V,E),L)$ satisfying
  $\abs{L(v)}\ge \deg[G]{v}+1$ for every $v\in V$, we have
  \[ \abs{P(G,L,v,i,D) - \Pr[G,L]{c(v)=i}} \le
    C\cdot\lambda^{D-3}, \]
  where $C>0$ is some constant.
\end{theorem}

We can view the one step recursion $P(G,L,v,i,D)$ as a function $F_i$
where each input of $F_i$ is obtained by calling a depth-$(D-1)$
recursion on some list-coloring instance $(G_k,L_k)$. Therefore $F_i$
has 2 main variations, depending on whether $P1$ or $P2$ is called.

It is natural to conceive of a sufficient condition that probably
looks like: the error of our estimation decays by a constant factor in
every iteration. However, this is not generally true even for systems
exhibiting correlation decay. This issue has already been addressed in
~\cite{li2012approximate,li2013correlation}, and in these works a
potential-based analysis is adopted. We will once more utilize this
method in our proof.

We choose
\[ \varphi(x)=2\ln x - 2\ln\tuple{\frac{1}{2}-x} \] whose derivative
(potential function) is
\[ \Phi(x) = \frac{1}{x(\frac{1}{2}-x)} \] and take
\[ M=\frac{3}{2}-\sqrt{2}=\sup_{0\le x\le
    \frac{1}{2}}\frac{1}{(1-x)\Phi(x)}. \]

Pick a list-coloring instance $(G=(V,E),L)$ with maximum
degree $3$, a vertex $v$ in $G$ with neighbor(s) $v_1$ and $v_2$ if
exist satisfying $\abs{L(v)}\ge \deg[G]{v}+2$ and a color $i$. To prove
Theorem~\ref{thm:cd}, the idea is to apply induction on $D$, which
can be formalized by the following lemma.
\begin{lemma}\label{lem:ind} Let $\lambda =\frac{9996}{10000}$ be a
  constant, then one of the following statements holds:
  \begin{enumerate}
  \item $F_i(\*x) = F_i(\*x^*) = 0$;
  \item $F_i(\*x) = F_i(\*x^*) = \frac{1}{2}$;
  \item
    $\abs{\varphi(F_i(\*x))-\varphi(F_i(\*x^*))} \le \lambda \cdot
    \max_{j:x_j \in
      \tuple{0,\frac{1}{2}}}\abs{\varphi(x_j)-\varphi(x^*_j)}$,
  \end{enumerate}
  where $\*x$ are the return values of subroutines called by
  $P(G,L,v,i,D)$ and $\*x^*$ are true values of those called
  instances.
\end{lemma}

We shall point out here if the first two cases do not occur then
$\varphi(F_i(\*x))$ and $\varphi(F_i(\*x^*))$ are always
well-defined. This is a simple corollary of Lemma~\ref{thm:consistency}. Instead of proving this lemma, we will
introduce Lemma~\ref{lem:meanvalue} which can directly imply Lemma~\ref{lem:ind}.

To ease the notation we first define the following. Let
$\varphi(\*x) = \tuple{\varphi(x_1), \varphi(x_2), \cdots,
  \varphi(x_d)}$ for any $d$-dimensional vector $\*x$,
$d \in \mathbb{N}$, and similarly define $\varphi^{-1}(\*x)$.

\begin{lemma} \label{lem:meanvalue} Suppose $d$ is the arity of
  $F_i$. Define the contraction rate
  \[ \alpha(\*x) =
    \sum_{j=1}^{d}\frac{\Phi(F_i(\*x))}{\Phi(x_j)}\abs{\pderiv{F_i(\*x)}{x_j}}.\]
  Then for all $\*x \in \text{Dom}(F_i) \subseteq [0,\frac{1}{2}]^d$, we have
  \[ \alpha(\*x) \le \lambda \] where $\lambda =\frac{9996}{10000}$.
\end{lemma}

Before delving into the proof, we first show that how to prove Lemma~\ref{lem:ind} by Lemma~\ref{lem:meanvalue}.

\begin{proof}[Proof of Lemma~\ref{lem:ind}]
  Let $\mathcal{I}$ be the index set of variables of $F_i$. Let
  $\*x_0 = \set{x_i \mid i \in \mathcal{I}, x_i \in \set{0,
      \frac{1}{2}}}$ and
  $\*x_1 = \set{x_i \mid i \in \mathcal{I}, x_i \in
    (0,\frac{1}{2})}$. Let $\mathcal{I}_0$ and $\mathcal{I}_1$ be the
  corresponding index set of $\*x_0$ and $\*x_1$. Define $\*x^*_0$ and
  $\*x^*_1$ similarly.
	
  Let $\*u_1 = \varphi(\*x_1)$, $\*u^*_1 = \varphi(\*x^*_1)$, and
  since $\varphi$ is strictly increasing we have
  $\*x_1=\varphi^{-1}(\*u_1)$ and
  $\*x^*_1=\varphi^{-1}(\*u^*_1)$. Notice $\*u^*_1$ is well-defined
  because we know $x_i \in (0,\frac{1}{2})$ if and only if
  $x^*_i \in (0,\frac{1}{2})$ by Lemma~\ref{thm:consistency}. In
  other words, $\*x_0$ and $\*x_1$ shares the same index set with
  $\*x^*_0$ and $\*x^*_1$, respectively.
	
  Introduce
  \begin{align*}
    g(t) = \varphi(F_i(\*x_0,\varphi^{-1}(t\*u_1+(1-t)\*u^*_1))
  \end{align*}
  so that
  $\varphi(F_i(\*x))-\varphi(F_i(\*x^*)) =
  \varphi(F_i(\*x_0,\*x_1))-\varphi(F_i(\*x^*_0,\*x^*_1)) =
  g(1)-g(0)$. By \emph{Mean Value Theorem} there exists
  $\tilde{t} \in (0,1)$ such that
  \begin{align*}
    \frac{g(1)-g(0)}{1-0}=g'(\tilde{t}).
  \end{align*}
  For convenience we denote
  $\tilde{\*u}_1=\tilde{t}\*u_1+(1-\tilde{t})\*u^*_1$ and
  $\tilde{\*x}_1=\varphi^{-1}(\tilde{\*u}_1)$. Clearly each component
  of $\tilde{\*x}_1$ lies between $0$ and $\frac{1}{2}$ since
  $\varphi$ is a monotone function. Simple derivative calculation
  yields
  \begin{align*}
    \abs{\varphi(F_i(\*x))-\varphi(F_i(\*x^*))} &= \abs{\sum_{j \in \mathcal{I}_1}\frac{\Phi(F_i(\*x_0,\tilde{\*x}_1))}{\Phi(\tilde{x}_j)}\pderiv{F_i(\*x_0,\tilde{\*x}_1)}{x_j}\cdot(u_j-u^*_j)} \\
                                                &\le \sum_{j \in \mathcal{I}_1}\frac{\Phi(F_i(\*x_0,\tilde{\*x}_1))}{\Phi(\tilde{x}_j)}\abs{\pderiv{F_i(\*x_0,\tilde{\*x}_1)}{x_j}}\cdot\abs{u_j-u^*_j} \\
                                                &\le \tuple{\sum_{j \in \mathcal{I}_1}\frac{\Phi(F_i(\*x_0,\tilde{\*x}_1))}{\Phi(\tilde{x}_j)}\abs{\pderiv{F_i(\*x_0,\tilde{\*x}_1)}{x_j}}} \cdot \max_{j \in \mathcal{I}_1}\abs{u_j-u^*_j}.
  \end{align*}
  Finally notice that if $x_j \in \set{0,\frac{1}{2}}$ then
  $\frac{1}{\Phi(x_j)} = 0$,
  \begin{align*}
    \sum_{j \in \mathcal{I}_1}\frac{\Phi(F_i(\*x_0,\tilde{\*x}_1))}{\Phi(\tilde{x}_j)}\abs{\pderiv{F_i(\*x_0,\tilde{\*x}_1)}{x_j}} &= \sum_{j \in \mathcal{I}_0}\frac{\Phi(F_i(\*x_0,\tilde{\*x}_1))}{\Phi(x_j)}\abs{\pderiv{F_i(\*x_0,\tilde{\*x}_1)}{x_j}} + \sum_{j \in \mathcal{I}_1}\frac{\Phi(F_i(\*x_0,\tilde{\*x}_1))}{\Phi(\tilde{x}_j)}\abs{\pderiv{F_i(\*x_0,\tilde{\*x}_1)}{x_j}} \\
                                                                                                                                   &\le \sup_{\*x \in [0,\frac{1}{2}]^d}\alpha(\*x) \\
                                                                                                                                   &\le \lambda.
  \end{align*}
  This completes the proof.
\end{proof}

We make some remarks. Here $F_i$ is just a general concept
representing the function of our algorithm. We use different
recursions to compute the marginal probability as the degrees of $v$
and its neighbors changes. As a consequence, the specific form,
including arity of $F_i$ has several variations, and depends on actual
situations. Moreover, in our analysis we will frequently refine the domain of $F_i$ because in some cases both true value and computed value never exceed a certain bound. Nevertheless, we can always obtain the expression of this
contraction rate $\alpha(\*x)$, and it turns out that we can bound
this rate for all variations of $F_i$.

The rest of this section is dedicated to prove Lemma~\ref{lem:meanvalue}. Our proof is based on the discussion on the
degree of $v$. Thanks to the symmetry between colors, we will only
need to prove for $i=1$. The proofs for other colors are identical.
	
\subsection{$\deg[G]{v} = 1$}\label{sec:deg=1}

Denote by $v_1$ the only neighbor of $v$. In this case $F_1$ has
three variations.
\begin{align*}
  F_1 = \begin{cases}
    \frac{1-x}{3} & |L(v)|=4 \\
    \frac{1-x}{2+y} & 1 \in L(v), j \notin L(v) \\
    0  & 1 \notin L(v)
  \end{cases}
\end{align*}

where $x=P(G_v,L_{1,1},v_1,1,D-1)$ and
$y=P(G_v,L_{1,j},v_1,j,D-1)$. We shall prove Lemma~\ref{lem:ind} for
the first two variations since the last one is trivial.

\begin{enumerate}
\item $|L(v)|=4$.
	
  The contraction rate writes as
  \[ \alpha(\*x) =
    \frac{\Phi(F_1(x))}{\Phi(x)}\abs{\pderiv{F_1(x)}{x}}. \] Moreover
  we have the following upper bound
  \[
    \frac{\Phi(F_1(x))}{\Phi(x)}\abs{\pderiv{F_1(x)}{x}} =
    \frac{1}{3}\cdot\frac{x\tuple{\frac{1}{2}-x}}{\frac{1-x}{3}\tuple{\frac{1}{2}-\frac{1-x}{3}}}
    = \frac{3x(1-2x)}{(1-x)(1+2x)} \le \frac{3x}{1+2x} \le \frac{3}{4}
    < 1.
  \]

\item $1 \in L(v), j \notin L(v)$.
	
  In this case $F_1$ is a binary function. The contraction rate writes
  as
  \begin{align*}
    \alpha(x,y) = \frac{\Phi(F_1(x,y))}{\Phi(x)}\abs{\pderiv{F_1(x,y)}{x}} + \frac{\Phi(F_1(x,y))}{\Phi(y)}\abs{\pderiv{F_1(x,y)}{y}}.
  \end{align*}
  We further discuss on three cases.
  \begin{enumerate}
  \item $1 \notin L(v_1)$ and $j \notin L(v_1)$.
		
    In this case $x$ and $y$ are accurately computed, hence no error occurs in our computation.
		 
  \item $1 \notin L(v_1)$ and $j \in L(v_1)$.
		
    Denote by $d_1$ degree of $v_1$ in graph $G_v$. Then
    $y=\frac{\prod_{k=1}^{d_1}(1-z_{jk})}{\sum_{l \in
        L(v_1)}\prod_{k=1}^{d_1}(1-z_{lk})} \ge
    \frac{\tuple{1-\frac{1}{2}}^2}{\tuple{1-\frac{1}{2}}^2+(1-0)
      \times 3}=\frac{1}{13}$. This lower bound also holds for $y^*$.
		
    If $1 \notin L(v_1)$, then $x=x^*=0$. Let $F_0=F_1(0,\cdot)$ be
    the function obtained by fixing $x=0$ in $F_1$. The contraction
    rate of $F_0$ is
    \begin{align*}
      \alpha(y) &= \frac{\Phi(F_0(y))}{\Phi(y)}\abs{\pderiv{F_0(y)}{y}} \\
                &=  \frac{y\tuple{\frac{1}{2}-y}}{\frac{1}{2+y}\tuple{\frac{1}{2}-\frac{1}{2+y}}} \cdot \frac{1}{(2+y)^2} \\
                &= (1-2y) \le \frac{11}{13}.
    \end{align*}
		
  \item $1 \in L(v_1)$.
		
    Similarly we could have $x,x^* \ge \frac{1}{13}$. Then
    \begin{align*}
      \alpha(x,y) &= \frac{\Phi(F_1(x,y))}{\Phi(x)}\abs{\pderiv{F_1(x,y)}{x}} + \frac{\Phi(F_1(x,y))}{\Phi(y)}\abs{\pderiv{F_1(x,y)}{y}} \\
      =& \frac{1}{\frac{1-x}{2+y}\tuple{\frac{1}{2}-\frac{1-x}{2+y}}}\tuple{\frac{x\tuple{\frac{1}{2}-x}}{2+y}+\frac{(1-x)y\tuple{\frac{1}{2}-y}}{(2+y)^2}} \\
      =& \frac{x\tuple{\frac{1}{2}-x}(2+y)+y\tuple{\frac{1}{2}-y}(1-x)}{(1-x)\tuple{x+\frac{y}{2}}}.
    \end{align*} 		
    We show that
    $\frac{x\tuple{\frac{1}{2}-x}(2+y)+y\tuple{\frac{1}{2}-y}(1-x)}{(1-x)\tuple{x+\frac{y}{2}}}\le\lambda$
    for $\lambda=\frac{9996}{10000}$, which is equivalent to
    \begin{equation}
      \label{eq:d1t}
      x\tuple{\frac{1}{2}-x}(2+y)+y\tuple{\frac{1}{2}-y}(1-x)\le \lambda\cdot (1-x)\tuple{x+\frac{y}{2}}.
    \end{equation}
    Inequality \eqref{eq:d1t} can be simplified to
    \begin{equation}
      \label{eq:d1t2}
      (1-x)y^2+\tp{x^2-\frac{\lambda}{2}x-\frac{1-\lambda}{2}}y+x^2-(1-\lambda)x\ge 0.
    \end{equation}
    Using the fact that $\frac{1}{13}\le x\le \frac{1}{2}$, we know
    the LHS of \eqref{eq:d1t2} is minimized at
    $y=\frac{2x^2-\lambda x-1+\lambda}{4x-4}$. Plugging this into
    \eqref{eq:d1t2} and it can be simplified to
    \[
      1 - 2\lambda + \lambda^2 + (16 - 14 \lambda - 2 \lambda^2) x +
      (-52 + 36 \lambda + \lambda^2) x^2 + (32 - 20\lambda) x^3 + 4
      x^4\le 0,
    \]
    which holds for $\frac{1}{13}\le x\le \frac{1}{2}$.
  \end{enumerate}
	
\end{enumerate}

To summarize the analysis in Section~\ref{sec:deg=1}, we have
\[ \alpha(\*x) \le \lambda = \frac{9996}{10000}. \]

\subsection{$\deg[G]{v} = 2$}\label{sec:deg=2}
Denote by $v_1, v_2$ two neighbors of $v$. Let $d_i = \deg[G_v]{v_i}$
and we have $d_1 \ge d_2$. In this case
\[ F_i = F_i(\*x,\*y) = \frac{(1-f_i)(1-y_i)}{\sum_{j \in
      L(v)}(1-f_j)(1-y_j)} \] where
\begin{align*}
  f_i = 
  \begin{cases}
    \frac{\prod_{k=1}^{d_1}(1-x_{k,i})}{\sum_{j \in L(v_1)}\prod_{k=1}^{d_1}(1-x_{k,j})} & i \in L(v_1) \\
    0 & i \notin L(v_1).
  \end{cases}
\end{align*}

\subsubsection{$d_1=2$}\label{sec:d1=2}

We first note that for $i,j\in L(v_1)$, $f_i/f_j$ is bounded by
constants.

\begin{proposition}\label{prop:ratio_bound}
  If $d_1=1$ or $2$ and for every $1\le k\le d_1$, $j\in L(v_1)$, we
  have $0\le x_{k,j}\le \frac{1}{2}$, then for every $i,j\in L(v_1)$,
  it holds that $\frac{1}{4}\le f_i/f_j\le 4$ and
  $f_i\ge\frac{1}{13}$.
\end{proposition}
\begin{proof}
  For every $i,j\in L(v_1)$, we have
  \[
    \frac{f_i}{f_j}=\frac{\prod_{k=1}^{d_1}(1-x_{k,i})}{\prod_{k=1}^{d_1}(1-x_{k,j})}.
  \]
  Then the bound for the ratio follows from $d_1=1,2$ and
  $0\le x_{k,j}\le \frac{1}{2}$ for every $1\le k\le d_1$,
  $j\in L(v_1)$.

  To see the lower bound for $f_i$, we note that $\abs{L(v)}\le 4$ and thus $1=\sum_{j\in L(v)}f_j\le f_i+4\sum_{j\in L(v)\setminus \set{i}}f_i\le 13f_i$.
\end{proof}

To prove Lemma~\ref{lem:meanvalue} it suffices to bound the
contraction rate
\begin{align*}
  \alpha(\*x,\*y) = \sum_{i=1}^{2}\sum_{j \in L(v_1)}\frac{\Phi(F_1)}{\Phi(x_{ji})}\abs{\pderiv{F_1(\*x)}{x_{ji}}}+\sum_{j=1}^{4}\frac{\Phi(F_1)}{\Phi(y_j)}\abs{\pderiv{F_1(\*y)}{y_j}}.
\end{align*}

Simple calculation yields
\begin{align*}
  &\sum_{i=1}^{2}\sum_{j \in L(v_1)}\frac{\Phi(F_1)}{\Phi(x_{ji})}\abs{\pderiv{F_1}{x_{ji}}} \\
  =& \sum_{i=1}^{2}\tuple{\frac{\Phi(F_1)}{\Phi(x_{1i})}\cdot\frac{F_1f_1}{1-x_{1i}}\cdot\sum_{k=2}^{4}\frac{F_k}{1-f_k} + \sum_{j \in L(v_1)\setminus \{1\}}\frac{\Phi(F_1)}{\Phi(x_{ji})}\cdot\frac{F_1f_j}{1-x_{ji}}\abs{\frac{1}{1-f_1}-\sum_{\substack{k=1 \\ k \neq j}}^{4}\frac{F_k}{1-f_k}}} \\
  \le& \sum_{i=1}^{2}\tuple{M\cdot\Phi(F_1)F_1\tuple{f_1\sum_{k=2}^{4}\frac{F_k}{1-f_k}+\sum_{j \in L(v_1)\setminus \{1\}}f_j\abs{\frac{1}{1-f_1}-\sum_{\substack{k=1 \\ k \neq j}}^{4}\frac{F_k}{1-f_k}}}} \defeq 2 \cdot P_1(\*f,\*y), 
\end{align*}
\begin{align*}
  \sum_{j=1}^{4}\frac{\Phi(F_1)}{\Phi(y_j)}\abs{\pderiv{F_1}{y_j}} = \frac{\Phi(F_1)}{\Phi(y_1)}\cdot\frac{F_1(1-F_1)}{1-y_1} + \sum_{j=2}^{4}\frac{\Phi(F_1)}{\Phi(y_j)}\cdot\frac{F_1F_j}{1-y_j} \defeq P_2(\*f,\*y).
\end{align*}

Now we only need to bound
\[ \alpha(\*x, \*y) = 2P_1(\*f,\*y)+P_2(\*f,\*y). \]

Notice that after substituting $M$ for $\frac{1}{(1-x)\Phi(x)}$ we can
ignore $\*x$ and treat $P_1$ and $P_2$ as functions of $\*f$ and
$\*y$, with some constraints on $\*f$ as we will see soon.

\paragraph{Discussion on the absolute value.}

Let
$D_j\defeq\frac{1}{1-f_1}-\sum_{\substack{k=1 \\k\ne
    j}}^{4}\frac{F_k}{1-f_k}$ for $j=2,3,4$. We show that at least two
of these $D_j$'s are nonnegative. Assume for the contraction that
$D_2,D_3<0$, then we obtain
\begin{align*}
  \frac{1}{1-f_1}-\frac{F_1}{1-f_1}-\frac{F_3}{1-f_3}-\frac{F_4}{1-f_4}&<0\\
  \frac{1}{1-f_1}-\frac{F_1}{1-f_1}-\frac{F_2}{1-f_2}-\frac{F_4}{1-f_4}&<0.
\end{align*}
This is equivalent to
\begin{align}
  (1-f_2)(1-y_2)+(f_1-f_3)(1-y_3)+(f_1-f_4)(1-y_4)&<0\label{eq:abs1}\\
  (1-f_3)(1-y_3)+(f_1-f_2)(1-y_2)+(f_1-f_4)(1-y_4)&<0\label{eq:abs2}
\end{align}
\eqref{eq:abs1}+\eqref{eq:abs2} gives
\[
  (1+f_1-2f_3)(1-y_3)+(1+f_1-2f_2)(1-y_2)+2(f_1-f_4)(1-y_4)<0
\]
Since $1+f_1-2f_3,1+f_1-2f_2>0$ and $0<y_2,y_3,y_4<\frac{1}{2}$,
\[
  3f_1+1-f_2-f_3-2f_4<0.
\]
Since $d_1=2$ we have $|L(v)|=4$ so Proposition~\ref{prop:ratio_bound} holds for all pairs of $f_i$, $f_j$, $1\le i < j\le 4$. Applying $f_1+f_2+f_3+f_4=1$, we obtain $4f_1<f_4$, which is a
contradiction.

\bigskip Therefore, we have either all $D_j$ for $j=2,3,4$ are
nonnegative or at most one of it is negative. Assume $D_2$ is
negative, i.e.,
\[
  (1-f_2)(1-y_2)+(f_1-f_3)(1-y_3)+(f_1-f_4)(1-y_4)<0.
\]
Since $(1-f_2)(1-y_2)\ge 0$, we have either $f_1<f_3$ or $f_1<f_4$ or
both. W.l.o.g. assume $f_1<f_3$, now we distinguish between two cases:
\begin{itemize}
\item ($f_1<f_4$) In this case, we can let $y_2=\frac{1}{2}$ and
  $y_3=y_4=0$, this gives
  \[
    1-f_2+2(f_1-f_3)+2(f_1-f_4)<0.
  \]
  Using the identity $f_1+f_2+f_3+f_4=1$, we obtain
  \[
    6f_1+f_2-1<0.
  \]
\item ($f_1\ge f_4$) In this case, we can let $y_2=\frac{1}{2}$,
  $y_3=0$ and $f_4=f_1$, this gives
  \[
    1-f_2+2(f_1-f_3)<0.
  \]
  Using $f_3=1-f_1-f_2-f_4\le\frac{7}{8}-f_1-f_2$, we obtain
  \[
    4f_1+f_2-\frac{3}{4}<0.
  \]
\end{itemize}

Now we can continue our analysis of $\alpha(\*x,\*y)$.
\paragraph{Case $1$: All $D_j$ are nonnegative for $j = 2, 3, 4$.}

Introduce the following function of $w$ and $f$
\[
  G_\xi(w,f)=\frac{1-f}{\Phi(1-\frac{w}{1-f})}+4M\xi\cdot\frac{w}{1-f} \]
where $\xi \in [0,1]$ is some constant parameter. The following two
lemmas would be very useful in our analysis.

\begin{lemma}\label{lem:jensen0} $G_0(w,f)$ is concave when
  $f \in [0,\frac{1}{2}]$ and $\frac{1-f}{2} \le w \le 1-f$, hence for
  all $w_i, f_i$ satisfying $f_i \in [0,\frac{1}{2}]$ and
  $\frac{1-f_i}{2} \le w_i \le 1-f_i, i=1,2,\cdots,n$, we have
  \[ \frac{G_0(w_1,f_1)+G_0(w_2,f_2)+\cdots+G_0(w_n,f_n)}{n} \le
    G_0\tp{\frac{w_1+w_2+\cdots+w_n}{n},\frac{f_1+f_2+\cdots+f_n}{n}}. \]
\end{lemma}
\begin{proof}
  The Hessian of $G_0(w,f)$ is $
  \begin{bmatrix}
    -\frac{2}{1-f} & -\frac{2w}{(1-f)^2}\\
    -\frac{2w}{(1-f)^2} & -\frac{2w^2}{(1-f)^3}
  \end{bmatrix}
  $, which is negative semi-definite when $f\in[0,\frac{1}{2}]$.
\end{proof}

\begin{lemma}\label{lem:jensen} For all
  $w_1, w_2, w_3 \in [0,\frac{1}{2}]$ and
  $f_1, f_2, f_3 \in [\frac{1}{13},\frac{1}{2}]$ such that
  $\frac{1-f_i}{2} \le w_i \le 1-f_i, i=1,2,3$, we have
  \begin{align*}
    \frac{1}{2}\tuple{G_\xi(w_1,f_1)+G_\xi(w_2,f_2)} &\le \kappa \cdot G_\xi\tuple{\frac{w_1+w_2}{2},\frac{f_1+f_2}{2}} \\
    \frac{1}{3}\tuple{G_\xi(w_1,f_1)+G_\xi(w_2,f_2)+G_\xi(w_3,f_3)} &\le \kappa \cdot G_\xi\tuple{\frac{w_1+w_2+w_3}{3},\frac{f_1+f_2+f_3}{3}}
  \end{align*}
  holds for any $\xi \in [0,1]$, where $\kappa=\frac{1038}{1000}$.
\end{lemma}
\begin{proof}
  
  First we shall point out that if the lemma holds for $\xi=1$, then it should hold for any other $0\le\xi<1$. \\
  
  Suppose the lemma holds for $\xi=1$. That is
  \[ \frac{1}{2}\tuple{G_1(w_1,f_1)+G_1(w_2,f_2)} \le \kappa \cdot
    G_1\tuple{\frac{w_1+w_2}{2},\frac{f_1+f_2}{2}}. \] Rewrite
  \[ G_\xi(w,f)=(1-\xi)G_0(w,f)+\xi G_1(w,f). \] Recall that $G_0$ is
  concave, thus
  \begin{align*}
    \frac{1}{2}\tuple{G_\xi(w_1,f_1)+G_\xi(w_2,f_2)} &\le (1-\xi)G_0\tuple{\frac{w_1+w_2}{2},\frac{f_1+f_2}{2}}+\xi\kappa\cdot G_1\tuple{\frac{w_1+w_2}{2},\frac{f_1+f_2}{2}} \\
                                                     &\le (1-\xi)\kappa\cdot G_0\tuple{\frac{w_1+w_2}{2},\frac{f_1+f_2}{2}}+\xi\kappa\cdot G_1\tuple{\frac{w_1+w_2}{2},\frac{f_1+f_2}{2}} \\
                                                     &= \kappa \cdot G_\xi\tuple{\frac{w_1+w_2}{2},\frac{f_1+f_2}{2}}.
  \end{align*}
  The same argument works for the 3-variable case. So it remains to prove the $\xi=1$ case. \\
  It can be rigorously proved by \emph{Mathematica} (the codes are in
  Section~\ref{sec:computer}) that for all
  $w_1, w_2, w_3 \in [0,\frac{1}{2}]$ and
  $f_1, f_2, f_3 \in [\frac{1}{13},\frac{1}{2}]$ such that
  $\frac{1-f_i}{2} \le w_i \le 1-f_i, i=1,2,3$, we have
  \[ \frac{1}{2}(G_1(w_1,f_1)+G_1(w_2,f_2)) \le \kappa_1\cdot
    G_1\tuple{\frac{w_1+w_2}{2},\frac{f_1+f_2}{2}} \]
  \[ \frac{1}{3}(G_1(w_1,f_1)+2G_1(w_2,f_2)) \le \kappa_2\cdot
    G_1\tuple{\frac{w_1+2w_2}{3},\frac{f_1+2f_2}{3}}. \] Here
  $\kappa_1=\frac{10195}{10000}, \kappa_2=\frac{10181}{10000}$ and
  $\kappa_1\kappa_2 \le \kappa$. As a consequence,
  \begin{align*}
    &\phantom{{}={}}\frac{1}{3}(G_1(w_1,f_1)+G_1(w_2,f_2)+G_1(w_3,f_3)\\
    &= \frac{1}{2}\tuple{\frac{1}{3}(2G_1(w_1,f_1)+G_1(w_2,f_2))+\frac{1}{3}(G_1(w_2,f_2)+2G_1(w_3,f_3))} \\
    &\le \kappa_2 \cdot \frac{1}{2}\tuple{G_1\tuple{\frac{2w_1+w_2}{3},\frac{2f_1+f_2}{3}}+G_1\tuple{\frac{w_2+2w_3}{3}+\frac{f_2+2f_3}{3}}} \\
    &\le \kappa_1\kappa_2 \cdot G_1\tuple{\frac{1}{2}
      \tuple{\frac{2w_1+w_2}{3}+\frac{w_2+2w_3}{3}},\frac{1}{2}
      \tuple{\frac{2f_1+f_2}{3}+\frac{f_2+2f_3}{3}}} \\
    &\le \kappa \cdot G_1\tuple{\frac{w_1+w_2+w_3}{3},\frac{f_1+f_2+f_3}{3}}.
  \end{align*}
	
\end{proof}

Recall that $f_j=0$ for $j \notin L(v_1)$, and
\begin{align*}
  &\phantom{{}={}}\sum_{j \in L(v_1) \setminus \set{1}}f_j\tuple{\frac{1}{1-f_1}-\sum_{\substack{k=1 \\ k \neq j}}^{4}\frac{F_k}{1-f_k}} \\
  &= \sum_{j=2}^{4}f_j\tuple{\frac{1}{1-f_1}-\sum_{\substack{k=1 \\ k \neq j}}^{4}\frac{F_k}{1-f_k}} 
  = 1 - \sum_{j=2}^{4}f_j\sum_{\substack{k=1 \\ k \neq j}}^{4}\frac{F_k}{1-f_k} \\
  &= 1 - \sum_{j=1}^{4}f_j\sum_{\substack{k=1 \\ k \neq j}}^{4}\frac{F_k}{1-f_k} + f_1\sum_{k=2}^{4}\frac{F_k}{1-f_k} 
  = 1 - \sum_{k=1}^{4}\frac{F_k}{1-f_k}\sum_{\substack{j=1 \\ j \neq k}}^{4}f_j + f_1\sum_{k=2}^{4}\frac{F_k}{1-f_k} \\
  &= f_1\sum_{j=2}^{4}\frac{F_j}{1-f_j}.
\end{align*}
So we have
\begin{align*}
  \alpha &= \Phi(F_1)F_1\tuple{\frac{1-F_1}{(1-y_1)\Phi(y_1)}+\sum_{j=2}^{4}\frac{F_j}{(1-y_j)\Phi(y_j)}+4Mf_1\sum_{j=2}^{4}\frac{F_j}{1-f_j}}.
\end{align*}
Define symmetric forms of $F_k$ as follows.
\[ \hat{F}_k(f_1,f_2,y_1,y_2) =
  \frac{(1-f_k)(1-y_k)}{(1-f_1)(1-y_1)+3(1-f_2)(1-y_2)}, \quad
  k=1,2. \] Then we can define the symmetric form of $\alpha$
\[ \hat{\alpha}(f_1,f_2,y_1,y_2) =
  \Phi(\hat{F}_1)\hat{F}_1\tuple{\frac{1-\hat{F}_1}{(1-y_1)\Phi(y_1)}+\frac{3\hat{F}_2}{(1-y_2)\Phi(y_2)}+12Mf_1\cdot\frac{\hat{F}_2}{1-f_2}}. \]
\begin{lemma}\label{lem:sym3+} For all
  $\*f,\*y \in [0,\frac{1}{2}]^4$ such that
  $\frac{1}{13} \le f_1,f_2,f_3,f_4 \le \frac{1}{2}$ and
  $f_1+f_2+f_3+f_4=1$, there exists
  $\hat{f}_2,\hat{y}_2 \in [0,\frac{1}{2}]$ such that
  $f_1+3\hat{f}_2=1$ and
  \[ \alpha(\*f,\*y) \le
    \kappa\cdot\hat{\alpha}(f_1,\hat{f}_2,y_1,\hat{y}_2) \] where
  $\kappa=\frac{1038}{1000}$.
\end{lemma}
\begin{proof}
  Let $w_k=(1-f_k)(1-y_k),k=2,3,4$,
  $A(\*f,\*y)=\sum_{j=1}^{4}(1-f_j)(1-y_j)$ be the denominator of
  $F_k$, and $\hat{A}(f_1,f_2,y_1,y_2)=(1-f_1)(1-y_1)+3(1-f_2)(1-y_2)$
  be the denominator of $\hat{F}_k$. Then
  \begin{align*}
    \alpha &= \Phi(F_1)F_1\tuple{\frac{1-F_1}{(1-y_1)\Phi(y_1)}+\frac{1}{A}\sum_{j=2}^{4}\frac{1-f_j}{\Phi(1-\frac{w_j}{1-f_j})}+4Mf_1\cdot\frac{w_j}{1-f_j}} \\
           &= \Phi(F_1)F_1\tuple{\frac{1-F_1}{(1-y_1)\Phi(y_1)}+\frac{1}{A}\sum_{j=2}^{4}G_{f_1}(w_j,f_j)}.
  \end{align*}
  Take $\hat{w}_2$ and $\hat{f}_2$ such that $3\hat{w}_2=w_1+w_2+w_3$,
  $3\hat{f}_2=f_1+f_2+f_3$, and take
  $\hat{y}_2=1-\frac{\hat{w}_2}{1-\hat{f}_2}$. Therefore
  $f_1+3\hat{f}_2=f_1+f_2+f_3+f_4=1$ and
  \[ A(\*f,\*y)=\hat{A}(f_1,\hat{f}_2,y_1,\hat{y}_2) \]
  \[ F_1(\*f,\*y) = \hat{F}_1(f_1,\hat{f}_2,y_1,\hat{y}_2). \]
  Furthermore, $w_j$ and $y_j$ satisfy the condition of Lemma~\ref{lem:jensen} hence
  \begin{align*}
    \alpha &\le \Phi(F_1)F_1\tuple{\frac{1-F_1}{(1-y_1)\Phi(y_1)} + 3\kappa \cdot \frac{G_{f_1}(\hat{w}_2,\hat{f}_2)}{A}} \\
           &= \Phi(\hat{F}_1)\hat{F}_1\tuple{\frac{1-\hat{F}_1}{(1-y_1)\Phi(y_1)} + 3\kappa\cdot\tuple{\frac{\hat{F}_2}{(1-\hat{y}_2)\Phi(\hat{y}_2)}+4Mf_1\cdot\frac{\hat{F}_2}{1-\hat{f}_2}}} \\
           &\le \kappa \cdot \Phi(\hat{F}_1)\hat{F_1}\tuple{\frac{1-\hat{F}_1}{(1-y_1)\Phi(y_1)}+\frac{3\hat{F}_2}{(1-\hat{y}_2)\Phi(\hat{y}_2)}+12Mf_1\cdot\frac{\hat{F}_2}{1-\hat{f}_2}} \\
           &= \kappa 
             \cdot \hat{\alpha}(f_1,\hat{f}_2,y_1,\hat{y}_2).
  \end{align*}
\end{proof}

\begin{lemma}\label{lem:resolve3+} For all
  $f_1,f_2,y_1,y_2 \in [0,\frac{1}{2}]$ such that
  $\frac{1}{13} \le f_1 \le \frac{1}{2}$ and $f_1+3f_2=1$, we have
  \[ \hat{\alpha}(f_1,f_2,y_1,y_2) \le \frac{963}{1000}. \]
\end{lemma}
\begin{proof}
  The lemma can be rigorously proved by \emph{Mathematica}. The codes
  are in Section~\ref{sec:computer}.
\end{proof}

Theorem~\ref{thm:consistency} and Proposition~\ref{prop:ratio_bound}
provide the condition for Theorem~\ref{lem:sym3+}, and combining
Theorem~\ref{lem:resolve3+} gives
\[ \alpha(\*f,\*y) \le \kappa \cdot
  \hat{\alpha}(f_1,\hat{f}_2,y_1,\hat{y}_2) \le \frac{1038}{1000}
  \cdot \frac{963}{1000} < \frac{9996}{10000}. \]

\paragraph{Case $2$: $D_j$ is negative for some $j$.}

Without loss of generality, we can assume $j = 2$, i.e.,
$\frac{1}{1-f_1}-\sum_{\substack{k=1 \\ k\neq 2}}^{4}\frac{F_k}{1-f_k}
< 0$. Therefore,
\[ \sum_{j=2}^{4}f_j\abs{\frac{1}{1-f_1}-\sum_{\substack{k=1 \\ k\neq
        j}}^{4}\frac{F_k}{1-f_k}} =
  f_1\sum_{j=2}^{4}\frac{F_j}{1-f_j}-2f_2\tuple{\frac{1}{1-f_1}-\sum_{\substack{k=1
        \\ k\neq 2}}^{4}\frac{F_k}{1-f_k}}. \] So we have
\[ \alpha =
  \Phi(F_1)F_1\tuple{\frac{1-F_1}{(1-y_1)\Phi(y_1)}+\sum_{j=2}^{4}\frac{F_j}{(1-y_j)\Phi(y_j)}+4Mf_1\sum_{j=2}^{4}\frac{F_j}{1-f_j}-4Mf_2\tuple{\frac{1}{1-f_1}-\sum_{\substack{k=1
          \\ k\neq 2}}^{4}\frac{F_k}{1-f_k}}} \] which is a function
of $\*f,\*y \in [0,1]^4$ where $f_1+f_2+f_3+f_4=1$.

Similarly, by exploiting the symmetry of $f_3$ and $f_4$, we define
the symmetric form of $F_1$.
\[ \hat{F}_1(f_1, f_2, f_3, y_1, y_2, y_3) =
  \frac{(1-f_1)(1-y_1)}{\hat{A}} \] where
\begin{align*}
  \hat{A} = (1-f_1)(1-y_1)+(1-f_2)(1-y_2)+2(1-f_3)(1-y_3).
\end{align*}
Then we can define the symmetric form of $\alpha$
\[ \hat{\alpha} =
  \frac{\Phi(\hat{F}_1)\hat{F}_1}{\hat{A}}\tuple{\hat{A}(1-\hat{F}_1)P_1+P_2+P_3} \]
where
\begin{align*}
  P_1 &= \frac{1}{(1-y_1)\Phi(y_1)}-\frac{4Mf_2}{1-f_1}, \\
  P_2 &= \frac{1-f_2}{\Phi(y_2)} + 4Mf_1(1-y_2), \\
  P_3 &= \frac{2(1-f_3)}{\Phi(y_3)} + 8M(f_1+f_2)(1-y_3).
\end{align*}
So $\hat{\alpha}$ is a function of $\*f,\*y \in [0,\frac{1}{2}]^3$.
\begin{lemma}\label{lem:sym2+1-} For all
  $\*f,\*y \in [0,\frac{1}{2}]^4$ such that
  $\frac{1}{13} \le f_3,f_4 \le \frac{1}{2}$ and $f_1+f_2+f_3+f_4=1$,
  there exists $\hat{f}_3,\hat{y}_3 \in [0,\frac{1}{2}]$ such that
  $f_1+f_2+2\hat{f}_3=1$ and
  \[ \alpha(\*f,\*y) \le \kappa \cdot
    \hat{\alpha}(f_1,f_2,\hat{f}_3,y_1,y_2,\hat{y}_3) \] where
  $\kappa=\frac{1038}{1000}$.
\end{lemma}
\begin{proof}
  Let $w_j = (1-f_j)(1-y_j)$ for $j=3,4$, and denote
  $A=A(w_1,w_2,w_3,w_4)=\sum_{j=1}^{4}w_j$ be the denominator of
  $F_k$. Then
  \begin{align*}
    \alpha &= \frac{\Phi(F_1)F_1}{A}\tuple{A(1-F_1)P_1+P_2+\sum_{j=3}^{4}\frac{1-f_j}{\Phi(y_j)}+4M(f_1+f_2)(1-y_j)} \\
           &= \frac{\Phi(F_1)F_1}{A}\tuple{A(1-F_1)P_1+P_2+\sum_{j=3}^{4}\frac{1-f_j}{\Phi(1-\frac{w_j}{1-f_j})}+4M(f_1+f_2)\cdot\frac{w_j}{1-f_j}} \\
           &= \frac{\Phi(F_1)F_1}{A}\tuple{A(1-F_1)P_1+P_2+\sum_{j=3}^{4}G_{f_1+f_2}(w_j,f_j)}.
  \end{align*}
  Take $\hat{w}_3$ and $\hat{f}_3$ such that
  $2\hat{w}_3=w_3+w_4, 2\hat{f}_3=f_3+f_4$, and take
  $\hat{y}_3=1-\frac{\hat{w}_3}{1-\hat{f}_3}$. Then we have
  $f_1+f_2+2\hat{f}_3=f_1+f_2+f_3+f_4=1$. Let
  $\hat{A}(w_1,w_2,w_3)=w_1+w_2+2w_3$, then clearly
  $A(w_1,w_2,w_3,w_4) = \hat{A}(w_1,w_2,\hat{w}_3)$. Since
  $f_1+f_2 \in [0,1]$ by Lemma~\ref{lem:jensen} we have
  \begin{align*}
    \alpha(\*f,\*y) &\le \frac{\Phi(F_1)F_1}{\hat{A}}\tuple{\hat{A}(1-\hat{F}_1)P_1+P_2+2G_{f_1+f_2}(\hat{w}_3,\hat{f}_3)} \\
                    &= \frac{\Phi(\hat{F}_1)\hat{F}_1}{\hat{A}}\tuple{\hat{A}(1-\hat{F}_1)P_1+P_2+2\kappa\cdot\tuple{\frac{1-\hat{f}_3}{\Phi(\hat{y}_3)}+4M(f_1+f_2)(1-\hat{y}_3)}} \\
                    &\le \kappa \cdot  \frac{\Phi(\hat{F}_1)\hat{F}_1}{\hat{A}}\tuple{\hat{A}(1-\hat{F}_1)P_1+P_2+\frac{2(1-\hat{f}_3)}{\Phi(\hat{y}_3)}+8M(f_1+f_2)(1-\hat{y}_3)} \\
                    &= \kappa \cdot \hat{\alpha}(f_1,f_2,\hat{f}_3,y_1,y_2,\hat{y}_3).
  \end{align*}
\end{proof}

\begin{lemma}\label{lem:resolve2+1-} For all
  $f_1,f_2,f_3,y_1,y_2,y_3 \in [0,\frac{1}{2}]$ satisfying
  \begin{align*}
    f_1+f_2+2f_3 &= 1, \\
    6f_1+f_2-1 &< 0, \\
    4f_1+f_2-\frac{3}{4} &< 0,
  \end{align*} 
  and
  \[ \frac{1}{13} \le f_1 \le \frac{1}{2}, \quad 0 \le f_2, f_3 \le
    \frac{1}{2}, \] we have
  \[ \hat{\alpha}(f_1,f_2,f_3,y_1,y_2,y_3) \le \frac{9163}{10000}. \]
\end{lemma}
\begin{proof}
  Recall that
  \[ \hat{\alpha} =
    \frac{\Phi(\hat{F}_1)\hat{F}_1}{\hat{A}}\tuple{\hat{A}(1-\hat{F}_1)P_1+P_2+P_3} \]
  where
  \begin{align*}
    P_1 &= \frac{1}{(1-y_1)\Phi(y_1)}-\frac{4Mf_2}{1-f_1}, \\
    P_2 &= \frac{1-f_2}{\Phi(y_2)} + 4Mf_1(1-y_2), \\
    P_3 &= \frac{2(1-f_3)}{\Phi(y_3)} + 8M(f_1+f_2)(1-y_3) \\
        &= \frac{1+f_1+f_2}{\Phi(y_3)}+8M(f_1+f_2)(1-y_3).
  \end{align*}
  Denote
  \[ A_1 = \hat{A}(1-\hat{F}_1) = (1-f_2)(1-y_2)+2(1-f_3)(1-y_3). \]
  So
  \[ \hat{\alpha} =
    \frac{2\tuple{A_1P_1+P_2+P_3}}{A_1-(1-y_1)(1-f_1)}. \] We
  substitute $P_1$ for
  $P_1^\prime = \frac{1}{(1-y_1)\Phi(y_1)}-\frac{4Mf_2}{1-1/13} \ge
  P_1$ and obtain an upper bound
  \[ \hat{\alpha} \le
    \frac{2(A_1P_1^\prime+P_2+P_3)}{A_1-(1-y_1)(1-f_1)}. \] Notice now
  both numerator and denominator are linear functions of
  $f_1$. Therefore it reaches the maximum value when $f_1$ is at its
  boundary. The next step is to let $f_1$ take its boundary values and
  simplify the formula.
  \begin{enumerate}
  \item $f_1 = \frac{1}{6}(1-f_2)$.
		
    \[ \alpha_1 =
      \frac{2(A_1P_1^\prime+P_2^\prime+P_3^\prime)}{A_1-(1-y_1)(1-\frac{1}{6}(1-f_2))} \]
    where
    \begin{align*}
      P_2^\prime &= \frac{1-f_2}{\Phi(y_2)} + \frac{2}{3}M(1-f_2)(1-y_2), \\
      P_3^\prime &= \frac{7+5f_2}{6\Phi(y_3)} + \frac{4}{3}M(5f_2+1)(1-y_3).
    \end{align*}
    
    It can be rigorously proved by \emph{Mathematica} that $\alpha_1 \le \frac{9138}{10000}$. The codes are in Section~\ref{sec:computer}.
		
  \item $f_1 = \frac{1}{4}\tuple{\frac{3}{4}-f_2}$.
		
    \[ \alpha_2 =
      \frac{2(A_1P_1^\prime+P_2^\prime+P_3^\prime)}{A_1-(1-y_1)(1-\frac{1}{4}(\frac{3}{4}-f_2))} \]
    where
    \begin{align*}
      P_2^\prime &= \frac{1-f_2}{\Phi(y_2)} + M\tuple{\frac{3}{4}-f_2}(1-y_2), \\
      P_3^\prime &= \frac{19+12f_2}{16\Phi(y_3)} + 6M\tuple{f_2+\frac{1}{4}}(1-y_3).
    \end{align*}
	
    It can be rigorously proved by \emph{Mathematica} that $\alpha_1 \le \frac{9163}{10000}$. The codes are in Section~\ref{sec:computer}.
		
  \item $f_1 = \frac{1}{13}$.
		
    \[ \alpha_3 =
      \frac{2(A_1P_1^\prime+P_2^\prime+P_3^\prime)}{A_1-(1-y_1)(1-\frac{1}{13})} \]
    where
    \begin{align*}
      P_2^\prime &= \frac{1-f_2}{\Phi(y_2)} + \frac{4}{13}M(1-y_2), \\
      P_3^\prime &= \frac{14+f_2}{13\Phi(y_3)} + 8M\tuple{f_2+\frac{1}{13}}(1-y_3).
    \end{align*}
    
    It can be rigorously proved by \emph{Mathematica} that $\alpha_3 \le \frac{9102}{10000}$.
    The codes are in Section~\ref{sec:computer}.
    
  \end{enumerate}
  
  To conclude we have
  $\hat{\alpha} \le \max\set{\frac{9138}{10000}, \frac{9163}{10000},
    \frac{9102}{10000}}=\frac{9163}{10000}$.
\end{proof}

The discussion of absolute values provides the condition for Lemma~\ref{lem:sym2+1-}, and combining Lemma~\ref{lem:resolve2+1-} gives
\[ \alpha(\*f,\*y) \le \kappa \cdot
  \hat{\alpha}(f_1,f_2,\hat{f}_3,y_1,y_2,\hat{y}_3) \le
  \frac{1038}{1000} \cdot \frac{9163}{10000} < \frac{9512}{10000}. \]

To summarize the analysis in Section~\ref{sec:d1=2}, we have
\[ \alpha(\*f,\*y) \le \max\set{\frac{9512}{10000},\frac{9996}{10000}}
  = \frac{9996}{10000}. \]

\subsubsection{$d_1=1$} \label{sec:d1=1}

When $d_1=1$ we need to bound
$\alpha(\*x,\*y)=P_1(\*f,\*y)+P_2(\*f,\*y)$. Furthermore, if
$1 \in L(v_1)$ then we still have
$\frac{1}{13} \le f_1 \le \frac{1}{2}, 0 \le f_2, f_3, f_4 \le
\frac{1}{2}$. In this case, the proof in Section~\ref{sec:d1=2} can
all go through once we obtain the symmetric form of $\alpha$ by the
following lemma. This is a modified version of Lemma~\ref{lem:jensen}
that can fit the situation of $d_1=1$.

\begin{lemma} \label{lem:jensen_mod} For all
  $w_1, w_2, w_3 \in [0,\frac{1}{2}]$ and
  $f_1, f_2, f_3 \in [0,\frac{1}{2}]$ such that
  $\frac{1-f_i}{2} \le w_i \le 1-f_i, i=1,2,3$, we have
  \begin{align*}
    \frac{1}{2}\tuple{G_\xi(w_1,f_1)+G_\xi(w_2,f_2)} &\le \kappa \cdot G_\xi\tuple{\frac{w_1+w_2}{2},\frac{f_1+f_2}{2}} \\
    \frac{1}{3}\tuple{G_\xi(w_1,f_1)+G_\xi(w_2,f_2)+G_\xi(w_3,f_3)} &\le \kappa \cdot G_\xi\tuple{\frac{w_1+w_2+w_3}{3},\frac{f_1+f_2+f_3}{3}}
  \end{align*}
  holds for any $\xi \in [0,\frac{1}{4}]$, where
  $\kappa=\frac{1019}{1000}$.
\end{lemma}
\begin{proof}
  The proof is almost the same as Lemma~\ref{lem:jensen}, except that
  here we only need to prove for $\xi=\frac{1}{4}$. This is also
  achieved by proving that for all $w_1, w_2, w_3 \in [0,\frac{1}{2}]$
  and $f_1, f_2, f_3 \in [0,\frac{1}{2}]$ such that
  $\frac{1-f_i}{2} \le w_i \le 1-f_i, i=1,2,3$, we have
  \[ \frac{1}{2}(G_\frac{1}{4}(w_1,f_1)+G_\frac{1}{4}(w_2,f_2)) \le
    \kappa_1\cdot
    G_\frac{1}{4}\tuple{\frac{w_1+w_2}{2},\frac{f_1+f_2}{2}} \]
  \[ \frac{1}{3}(G_\frac{1}{4}(w_1,f_1)+2G_\frac{1}{4}(w_2,f_2)) \le
    \kappa_2\cdot
    G_\frac{1}{4}\tuple{\frac{w_1+2w_2}{3},\frac{f_1+2f_2}{3}} \]
  where $\kappa_1=\frac{1009}{1000}, \kappa_2=\frac{1009}{1000}$ and
  $\kappa_1\kappa_2 \le \kappa$. The \emph{Mathematica} code to verify
  the lemma is in Section~\ref{sec:computer}.
\end{proof}

Now it remains to handle the case when $1 \notin L(v_1)$. So in the
rest of this section we will assume that $ f_1 = 0$.

According to the convention in Algorithm~\ref{algo:marg-deg2}, we have
either $1 \notin L(v_2)$ or $d_2=0$. We will defer the discussion of
this $d_2=0$ case to the end of this section. If $1 \notin L(v_2)$ we
have $f_1=y_1=0$ and this is true for both actual value and computed
value. So we fix $f_1, y_1$ to be zero in our recursion and discuss
the contraction rate of this partially fixed function
\[ F_1 = \frac{1}{1+\sum_{j=2}^{4}(1-f_j)(1-y_j)}. \] The contraction
rate should not involve the derivatives of $f_1$ and $y_1$, namely
\[ \alpha(\*f,\*y) = P_1(\*f,\*y)+P_2(\*f,\*y) \] where
\begin{align*}
  P_1(\*f,\*y) &= \Phi(F_1)F_1 \cdot M \cdot \sum_{j \in L(v_1)\setminus \set{1}}f_j\abs{1-\sum_{\substack{k=1 \\ k\neq j}}^{4}\frac{F_k}{1-f_k}}, \\
  P_2(\*f,\*y) &= \Phi(F_1)F_1 \cdot \sum_{j=2}^{4}\frac{F_j}{(1-y_j)\Phi(y_j)},
\end{align*}
and
\[ F_k = \frac{(1-f_k)(1-y_k)}{1+\sum_{j=2}^{4}(1-f_j)(1-y_j)} \] is
also partially fixed accordingly.

\paragraph{Discussion on the absolute values.}
Let $D_j\defeq1-\sum_{\substack{k=1 \\k\ne j}}^{4}\frac{F_k}{1-f_k}$
for $j=2,3,4$. Recall that
\begin{align*}
  \sum_{j=2}^{4}D_j = \sum_{j \in L(v_1) \setminus \set{1}}f_j\tuple{1-\sum_{\substack{k=1 \\ k \neq j}}^{4}\frac{F_k}{1-f_k}} = f_1\sum_{j=2}^{4}\frac{F_j}{1-f_j} = 0,
\end{align*}
so it cannot be the case that all $D_j$'s have the same sign. We will
always, without loss of generality, assume $D_2$ has the opposite sign
against others. Then $|D_2|+|D_3|+|D_4|$ is either $2D_2$ or $-2D_2$.

\paragraph{Case 1: $D_2$ is negative.}

In this case
\[ \alpha(\*f,\*y) =
  \Phi(F_1)F_1\tp{M\cdot(-2D_2)+\sum_{j=2}^{4}\frac{F_j}{(1-y_j)\Phi(y_j)}}. \]
Denote $A\defeq 1+\sum_{j=2}^{4}(1-f_j)(1-y_j)$ the denominator of
$F_1$.

We first consider the case when $y_j=\frac{1}{2}$ for some
$j \in \set{2,3,4}$. By Theorem~\ref{thm:bounds} we know that all
$y_j$'s should be accurately computed given the recursion depth $D$ is
at least $3$. So we can further discard all derivatives of $y_j$ and
obtain
\begin{align*}
  \alpha(\*f,\*y) &= 2Mf_2 \cdot \Phi(F_1)F_1\tp{\sum_{\substack{k=1 \\ k\neq 2}}^{4}\frac{F_k}{1-f_k}-1} \\
                  &= \frac{4Mf_2(3-y_3-y_4-A)}{A-2}.
\end{align*} 
Notice that $\alpha$ is monotonically increasing on $y_2$, so we take
$y_2=\frac{1}{2}$. After substituting $1-f_2$ for $f_3+f_4$ we get
\[ \alpha(\*f,\*y) \le 4Mf_2 \cdot
  \frac{f_3(\frac{1}{2}-y_3)+f_4(\frac{1}{2}-y_4)}{(\frac{1}{2}-y_3)(1-f_3)+(\frac{1}{2}-y_4)(1-f_4)}
  \le 2M \] where the last inequality is due to
$f_2,f_3,f_4 \le \frac{1}{2}$ and the monotonicity on $f_3$ and $f_4$.

On the other aspect, if $y_j \neq \frac{1}{2}$ for all
$j \in \set{2,3,4}$, then by Theorem~\ref{thm:bounds} we have
$y_j \le \frac{6}{13}$ for all $j \in \set{2,3,4}$ since $d_2$ is at
most 1. Let $w_j=(1-f_j)(1-y_j)$, by Lemma~\ref{lem:jensen_mod}

\begin{align*}
  \alpha(\*f,\*y) &= \Phi(F_1)F_1\tp{2Mf_2\tp{\sum_{\substack{k=1 \\ k\neq 2}}^{4}\frac{F_k}{1-f_k}-1}+\sum_{j=2}^{4}\frac{F_j}{(1-y_j)\Phi(y_j)}} \\
                  &= \frac{\Phi(F_1)F_1}{A}\tp{\frac{1-f_2}{\Phi(y_2)}+2Mf_2(1-A)+\sum_{j=3}^{4}\frac{1-f_j}{\Phi(y_j)}+2Mf_2(1-y_j)} \\
                  &= \frac{\Phi(F_1)F_1}{A}\tp{\frac{1-f_2}{\Phi(y_2)}+2Mf_2(1-A)+\sum_{j=3}^{4}\frac{1-f_j}{\Phi(1-\frac{w_j}{1-f_j})}+2Mf_2\frac{w_j}{1-f_j}} \\
                  &= \frac{\Phi(F_1)F_1}{A}\tp{\frac{1-f_2}{\Phi(y_2)}+2Mf_2(1-A)+\sum_{j=3}^{4}G_{\frac{f_2}{2}}(w_j,f_j)} \\
                  &\le \kappa \cdot \frac{\Phi(\hat{F}_1)\hat{F}_1}{\hat{A}}\tp{\frac{1-f_2}{\Phi(y_2)}+2Mf_2(1-\hat{A})+2G_{\frac{f_2}{2}}(\hat{w}_3,\hat{f}_3)}
\end{align*}
where $\hat{w}_3 = \frac{w_3+w_4}{2}$,
$\hat{f}_3 = \frac{f_3+f_4}{2}$,
$\hat{F}_1=\frac{1}{1+w_2+2\hat{w}_3}$ and
$\hat{A}=1+w_2+2\hat{w}_3$. If we take
$\hat{y}_3 = 1-\frac{\hat{w}_3}{1-\hat{f}_3}$ then we can get the
symmetric form of $\alpha$:
\[ \hat{\alpha}(\*f,\*y) =
  \frac{\Phi(\hat{F}_1)\hat{F}_1}{\hat{A}}\tp{\frac{1-f_2}{\Phi(y_2)}+2Mf_2(1-\hat{A})+\frac{2(1-\hat{f}_3)}{\Phi(\hat{y}_3)}+4Mf_2(1-\hat{y}_3)}. \]

\begin{lemma} \label{lem:resolve2+1-d1} For all
  $f_2,f_3,y_2,y_3 \in [0,\frac{1}{2}]$ satisfying
  \begin{align*}
    f_2+2f_3 = 1, \\
    \frac{1}{13} \le f_2 \le \frac{1}{2}, \\
    0 \le y_2, y_3 \le \frac{6}{13},
  \end{align*} 
  we have
  \[ \hat{\alpha}(f_2,f_3,y_2,y_3) \le \frac{9231}{10000}. \]
\end{lemma}
\begin{proof}
  The lemma can be verified by \emph{Mathematica}. The codes are in
  Section~\ref{sec:computer}.
\end{proof}

In conclusion we have
\[ \alpha(\*f,\*y) \le \max\set{2M, \kappa \cdot
    \hat{\alpha}(f_2,\hat{f}_3,y_2,\hat{y}_3)} \le \max\set{2M,
    \frac{1018}{1000} \cdot \frac{9231}{10000}} < \frac{94}{100}. \]

\subsubsection{Case 2: $D_2$ is positive}

In this case
\[ \alpha(\*f,\*y) = \Phi(F_1)F_1\tp{M\cdot 2D_2 +
    \sum_{j=2}^{4}\frac{F_j}{(1-y_j)\Phi(y_j)}}. \]

As we did in Case 1, we first consider when $y_j=\frac{1}{2}$ for some
$j \in \set{2,3,4}$. We similarly obtain
\begin{align*}
  \alpha(\*f,\*y) &= 2Mf_2 \cdot \Phi(F_1)F_1\tp{1-\sum_{\substack{k=1 \\ k\neq 2}}^{4}\frac{F_k}{1-f_k}} \\
                  &= \frac{4Mf_2(A-3+y_3+y_4)}{A-2}.
\end{align*} 

Notice that $\alpha(\*f,\*y)$ is monotonically increasing on $y_3$ and
$y_4$, so we take $y_3=y_4=\frac{1}{2}$ which yields
\[ \alpha(\*f,\*y) \le 4Mf_2 \le 2M. \] Now we once more assume
$y_j \le \frac{6}{13}$ for all $j \in \set{2,3,4}$. Recall
$\lambda=\frac{9996}{10000}$, we now prove that
\begin{align*}
  \alpha(\*f,\*y)&=\frac{\sum_{j=2}^4(1-f_j)y_j(\frac{1}{2}-y_j)+2Mf_2(A-3+y_3+y_4)}{\frac{A}{2}-1}< \lambda.
\end{align*}

Since the denominator of $\alpha(\*f,\*y)$ is positive,
$\alpha(\*f,\*y)< \lambda$ is equivalent to
\[
  G\defeq\sum_{j=2}^4(1-f_j)y_j(\frac{1}{2}-y_j)+2Mf_2\tp{A-3+y_3+y_4}-\lambda\tp{\frac{1}{2}A-1}<0.
\]
Note that $G$ is quadratic on $y_3$, we can write it as
\begin{align*}
  G
  &=-(1-f_3)y_3^2+\tp{2Mf_2 + \frac{1}{2}(1-f_3)+\frac{1}{2}\lambda(1-f_3)- 2Mf_2(1-f_3)}y_3+C\\
  &=(1-f_3)\tp{-y_3^2+\tp{\frac{2Mf_2}{1-f_3}+\frac{1+\lambda}{2}-2Mf_2}y_3}+C,
\end{align*}
where $C$ is a polynomial containing no $y_3$.

Therefore, $G$ is increasing in $[-\infty,x_0]$ where
$x_0=\frac{Mf_2}{1-f_3}+\frac{1+\lambda}{4}-Mf_2\ge\frac{1+\lambda}{4}\ge\frac{6}{13}$. Since
$y_3$ and $y_4$ are symmetric, the same argument holds for $y_4$.

We only need to prove that $ G'\defeq
G|_{y_3=y_4=\frac{6}{13}}<0$. Applying $f_2+f_3+f_4=1$, a direct
calculation yields
\begin{align*}
  G'&=\frac{2}{13}Mf_2^2(-6 + 13 y_2) + \frac{1}{338} (6 - 91\lambda + 169 y_2 + 169\lambda y_2- 
      338y_2^2)\\
    &\quad +\frac{1}{338} f_2\tp{-13 \lambda (-6 + 13 y_2) + (-6 + 13 y_2) (-1 - 52 M + 26 y_2)}.
\end{align*}
Since $y_2\le \frac{6}{13}$, $G'$ is increasing in $[-\infty,x_1]$
where $x_1=\frac{1 + 52 M + 13 \lambda - 26 y_2}{104 M}>
\frac{1}{2}$. Therefore, we only need to prove that
\[
  G''\defeq G'|_{f_2=\frac{1}{2}}=\frac{9}{338} + \frac{3 M}{13} -
  \frac{2 \lambda}{13} + \tp{\frac{1}{4} - \frac{M}{2} +
    \frac{\lambda}{4}}y_2 - \frac{y_2^2}{2}<0,
\]
which holds for $y_2\in\left[0,\frac{6}{13}\right]$.

In conclusion we have
\[ \alpha(\*f,\*y) \le \max\set{2M,\lambda} = \lambda. \]

\paragraph{The case $d_2=0$.} 

At last we come to the discussion for $d_2=0$. In this case $y_1$ is not necessarily 0, but all $y_j$'s are accurately computed. Redefine
\begin{align*}
F_1 &= \frac{1-y_1}{1-y_1+\sum_{j=2}^{4}(1-f_j)(1-y_j)}, \\
A &= 1-y_1+\sum_{j=2}^{4}(1-f_j)(1-y_j).
\end{align*}
As we did before, we shall discard the derivatives of $y_j$'s and assume $D_2$ has the opposite sign against others. Now
\begin{align*}
\alpha(\*f,\*y) &= 2Mf_2 \cdot \Phi(F_1)F_1\abs{1-\sum_{\substack{k=1 \\ k\neq 2}}^{4}\frac{F_k}{1-f_k}} \\
 &= \frac{4Mf_2\abs{A-(1-y_1)-(1-y_3)-(1-y_4)}}{A-2(1-y_1)} \\
 &= \frac{4Mf_2\abs{\sum_{j=2}^{4}(1-f_j)(1-y_j)-2+y_3+y_4}}{\sum_{j=2}^{4}(1-f_j)(1-y_j)-(1-y_1)}
\end{align*}
is monotonically decreasing on $y_1$. So we can take $y_1=0$ and this is reduced to a situation we have discussed before.

To summarize the analysis in Section~\ref{sec:d1=1}, we have
\[ \alpha(\*f,\*y) \le \max\set{\frac{94}{100},\lambda} = \lambda. \]

So far we have exhausted all possible cases when $\deg[G]{v}=2$. Putting together the conclusions of Section~\ref{sec:deg=1} and Section~\ref{sec:deg=2}, we can finish the proof of Lemma~\ref{lem:meanvalue}.

\subsection{Proof of Theorem~\ref{thm:cd}}

By the discussion on cases in~\ref{sec:deg=1} and \ref{sec:deg=2}, we
have finished the proof of Lemma~\ref{lem:meanvalue} so far.

Thus we can prove Theorem~\ref{thm:cd} now.

\begin{proof}[Proof of Theorem~\ref{thm:cd}]
  Let $\lambda = \frac{9996}{10000}$ be constant.
  
  We first claim that if a vertex $v$ satisfies $\deg[G]{v}\le 2$ and
  $\abs{L(v)}\ge \deg[G]{v}+2$, then one of the following statements
  holds:
  \begin{itemize}
  \item $P(G,L,v,i,D)=\Pr[G,L]{c(v)=i}$;
  \item
    $\abs{\varphi(P(G,L,v,i,D))-\varphi(\Pr[G,L]{c(v)=i})}\le C_1\cdot\lambda^{D-2}$,
    where $\varphi(x)=2\ln x-2\ln\tp{\frac{1}{2}-x}$ and $C_1>0$ is a constant.
  \end{itemize}
  
  Given the claim, we have for some constant $C_2>0$, it holds that
  \[
    \abs{P(G,L,v,i,D)-\Pr[G,L]{c(v)=i}}= \frac{1}{\Phi(\tilde
      x)}\cdot\abs{\varphi(P(G,L,v,i,D))-\varphi(\Pr[G,L]{c(v)=i})}\le C_2\cdot\lambda^D,
  \]
  where $\Phi(x)\defeq\varphi'(x)=\frac{1}{x\tp{\frac{1}{2}-x}}$ and
  $\tilde x$ is some real between $\varphi(P(G,L,v,i,D))$ and
  $\varphi(\Pr[G,L]{c(v)=i})$.

  Now assume $\tp{G=(V,E),L}$ satisfies $\abs{L(v)}\ge\deg[G]{v}+1$
  for every $v\in V$. Let $v\in V$ be an arbitrary vertex and consider
  the computation tree of $P(G,L,v,i,D)$. According to the
  construction in Section~\ref{sec:prelim}, all the smaller instances
  $P(G',L',v',i',D')$ called by the procedure satisfy
  $\abs{L(v)}\ge \deg[G']{v'}+2$ and $\deg[G']{v'}\le 2$, i.e., the
  condition specified in the above claim. Further note that in all
  cases, the 1-norm of the gradients of our recursions
  \begin{align*}
    F(x,y,z)&=\frac{1-x}{3-x-y} & \mbox{if $\deg[G]{v}=1$ and $\abs{L(v)}=2$;}\\
    F(x,y) &=\frac{1-x}{2+y} & \mbox{if $\deg[G]{v}=1$ and $\abs{L(v)}=3$;}\\
    F(x) &=\frac{1-x}{3} & \mbox{if $\deg[G]{v}=1$ and $\abs{L(v)}=4$;}\\
    F(\*f,\*y) &=\frac{(1-f_i)(1-y_i)}{\sum_{j\in L(v)}(1-f_j)(1-y_j)} & \mbox{if $\deg[G]{v}=2$;}\\
    F(\*x,\*y,\*z) &=\frac{(1-x_i)(1-y_i)(1-z_i)}{\sum_{j\in L(v)}(1-x_j)(1-y_j)(1-z_j)} & \mbox{if $\deg[G]{v}=3$,}
  \end{align*}
  are bounded above by some constants for parameters in the range
  $[0,\frac{1}{2}]$. Therefore it follows from the mean value theorem
  and the claim that
  \[
    \abs{P(G,L,v,i,D)-\Pr[G,L]{c(v)=i}}\le
    C\cdot\lambda^D.
  \]
  for some constant $C>0$.
  
  It remains to prove the claim. We apply induction on $D$. The base
  case is that $D=2$. It follows from Theorem~\ref{thm:bounds} and
  Lemma~\ref{thm:consistency} that if $\Pr[G,L]{c(v)=i}$ is $0$ or
  $\frac{1}{2}$, then the algorithm return the correct value, i.e.,
  $P(G,L,v,i,D)=\Pr[G,L]{c(v)=i}$. Otherwise, the function
  $\varphi(\cdot)$ is bounded from above and thus the claim holds.
  For $D>2$, the claim follows from the induction hypothesis and
  Lemma~\ref{lem:ind}.
\end{proof}


\section{Proof of the Main Theorem}\label{sec:proofmain}

In this section, we prove Theorem~\ref{thm:main}. We start the proof
by first analyzing the running time of Algorithm~\ref{algo:marg-main}.

Let $G=(V,E)$ be a graph with $\abs{V}=n$, $L$ be its color lists,
$v\in V$ be a vertex, $i\in\set{1,2,3,4}$ be a color and $D$ be
nonnegative integer. Let $\tau(G,L,v,i,D)$ denote the running time of
the procedure $P(G,L,v,i,D)$, then we have:

\begin{lemma}\label{lem:time-pre}
  If $\deg[G]{v}\le 2$, then $\tau(G,L,v,i,D)=O\tp{n^312^D}$.
\end{lemma}
\begin{proof}
  We apply induction on $n$ to show that for some constant $C\ge 0$,
  $\tau(G,L,v,i,D)\le C\cdot n^312^D$. The base case is that $n=1$,
  then the algorithm terminates in constant time.

  For general $n$, we need to analyze cases $\deg[G]{v}=1,2$
  respectively.

  \textbf{Case $\deg[G]{v}=1$:} Algorithm~\ref{algo:marg-deg1}
  contains two subcases. We use an adjacency matrix to represent a
  graph.  Thus we can construct in $n^2$ time the graph $G_v$ which
  contains $n-1$ vertices.  We then have the following recursions for
  the two cases respectively (assuming notations in the description of
  Algorithm~\ref{algo:marg-deg1}):
  \begin{align*}
    \tau(G,L,v,i,D)&\le \tau(G_v,L_{1,i},v_1,i,D-1)+n^2\\
    \tau(G,L,v,i,D)&\le \tau(G_v,L_{1,i},v_1,i,D-1)+\tau(G_v,L_{1,j},v_1,j,D-1)+n^2
  \end{align*}
  Then the lemma follows from the induction hypothesis.

  \textbf{Case $\deg[G]{v}=2$:} Algorithm~\ref{algo:marg-deg2} has at
  most 12 branches, we have (assuming notations in the description of
  Algorithm~\ref{algo:marg-deg2}):
  \[
    \tau(G,L,v,i,D)\le \sum_{k\in d_1}\sum_{w\in
      L(v_1)}\tau(G_{v,v_1},L_{k,w}',w,D-1)+\sum_{j\in
      L(v)}\tau(G_v,L_{2,j},j,D-1)+n^2.
  \]
  Then the lemma follows from the induction hypothesis.
\end{proof}

If $\deg[G]{v}=3$, then the algorithm $P(G,L,v,i,D)$ will call
$P3(G,L,v,i,D)$ described in Algorithm~\ref{algo:marg-deg3}. However,
since the maximum degree of $G$ is at most three hence in further recursion call to Algorithm~\ref{algo:marg-main}, the degree
of a vertex decreases by at least one. Therefore
Algorithm~\ref{algo:marg-deg3} can be called at most once. Combining
Lemma~\ref{lem:time-pre}, we have

\begin{lemma}\label{lem:time}
  $\tau(G,L,v,i,D)=O\tp{n^312^D}$.
\end{lemma}

Now we prove Lemma~\ref{lem:approx}.

\begin{proof}[Proof of Lemma~\ref{lem:approx}]
   First, we need to bound the value $\Pr[G,L]{c(v)=i}$ on the computation tree. If $\Pr[G,L]{c(v)=i}=0$ then it is clear to see $P(G,L,v,i,D)=0$ thus we are done. Otherwise we have $\Pr[G, L]{c(v) = i} \ge \frac{1}{13}$ if $(G, L, v)$ is a \emph{reachable} instance. In previous discussion we know that $(G, L, v)$ is on the root of our computation tree if this instance is not \emph{reachable}. In this case
  \[ \Pr[G,L]{c(v)=i} = \frac{\prod_{k=1}^{d}(1-\Pr[G_v,L_{k,i}]{c(v_k)=i})}{\sum_{j\in L(v)} \prod_{k=1}^{d}(1-\Pr[G_v,L_{k,j}]{c(v_k)=j})}, \]
  where $d=\deg[G]{v}\le 3$ and $|L(v)|\le 4$. It yields 
  \[ \Pr[G,L]{c(v)=i} \ge \frac{\tp{1-\frac{1}{2}}^3}{\tp{1-\frac{1}{2}}^3 +1+ 1+1} = \frac{1}{25}. \]
  Combining with the bound of reachable cases it implies $\Pr[G,L]{c(v)=i}\ge \frac{1}{25}$ for all instances in computation tree.

By Theorem~\ref{thm:cd}, there exists constants
$\lambda =\frac{9996}{10000}$ and $C > 0$ such that for every list-coloring instance $(G,L)$ satisfying conditions in the
statement of the theorem, it holds that
\[ \abs{P(G,L,v,i,D) - \Pr[G,L]{c(v)=i}} \le C\cdot
\lambda^{D-3} \] for all $D \ge 3$. 

  For any $0 < \eps < 1$, let $t$ be the smallest integer such that
  $C\cdot\lambda^{t-3} \le \frac{\eps}{25}$ and
  let $\hat p = P(G, L, v, i, t)$.
  We can show that Algorithm~\ref{algo:marg-main} up to depth $t$ is the algorithm outputs $\hat{p}$ such that 
  \[(1-\eps)\hat p  \le\Pr[G,L]{c(v)=i} \le (1+\eps)\hat p\] in time $\mathrm{poly}(\abs{V},\frac{1}{\eps})$.
  
  Theorem~\ref{thm:cd} implies
  \[\Pr[G,L]{c(v)=i} - \frac{\eps}{25} \le \hat p \le \Pr[G,L]{c(v)=i} +
  \frac{\eps}{25}\]
  and thus by the bound of $\Pr[G,L]{c(v)=i}$ above it holds that
  \[
    (1-\eps)\Pr[G,L]{c(v)=i} \le\hat p \le (1+\eps)\Pr[G,L]{c(v)=i}\,.
  \]
  So
  \[
   (1-\eps)\hat p \le \frac{1}{1 + \eps} \hat p \le \Pr[G,L]{c(v)=i} \le \frac{1}{1 - \eps} \hat p \le (1 + \eps) \hat p\,.
  \]
  
  Next we show that Algorithm~\ref{algo:marg-main} up to depth $t$ is a
  polynomial time algorithm with respect to $\abs{V}$ and
  $\frac{1}{\eps}$. By Lemma~\ref{lem:time},
  $\tau(G,L,v,x,t)=O\tp{n^312^t}$. Since $t$ is the smallest integer
  such that
  $C\cdot\lambda^{t-3} \le \frac{\eps}{25}$, we
  have
  \[t - 4 \le \log_\lambda{\frac{\eps}{25C}} \le t - 3\,,\] which
  implies
  $\tau(G,L,v,x,t)=O\tp{n^312^{\log_\lambda{\frac{\eps}{25C}}}}
  =
  O\tp{n^3\left(\frac{25C}{\eps}\right)^{-\log_\lambda
      {12}}}$.  $\lambda$ and $C$
  are constants, so
  $\tau(G,L,v,x,t)=\mathrm{poly}\!\tp{\abs{V},\frac{1}{\eps}}$.
\end{proof}

Finally, combining Lemma~\ref{lem:approx} and Lemma~\ref{lem:margtopartition} completes the proof of Theorem~\ref{thm:main}.


\section{Computer Assisted Proofs}\label{sec:computer}

We use some \emph{Mathematica} codes to assist our proof. All the
codes are summaried in this section.

\paragraph{Initialization} We use following code to initilize our
computer assisted part of the proof.
\begin{lstlisting}
  Phi[x_] := 1/x/(1/2 - x);
  M = (Maximize[{1/Phi[x]/(1-x), 0<=x<=1/2}, {x}])[[1]];
  F[k_] := (1-f[k])(1-y[k])/Sum[(1-f[i])(1-y[i]), {i,1,4}];
\end{lstlisting}
\paragraph{Code for Lemma~\ref{lem:jensen}} ~
\begin{lstlisting}
  G[w_,f_] := (1-f)/Phi[1-w/(1-f)] + 4M w/(1-f);
  Resolve[Exists[{w1,w2,f1,f2}, (G[w1,f1]+G[w2,f2])>=10195/10000(2G[(w1+w2)/2,(f1+f2)/2]) && 1/13<=f1<=1/2 && 1/13<=f2<=1/2 && 1/2<=w1/(1-f1)<=1 && 1/2<=w2/(1-f2)<=1]]
  Resolve[Exists[{w1,w2,f1,f2}, (G[w1,f1]/3 + 2G[w2,f2]/3)>=10181/10000(G[(w1+2w2)/3,(f1+2f2)/3]) && 1/13<=f1<=1/2 && 1/13<=f2<=1/2 && 1/2<=w1/(1-f1)<= 1 && 1/2<=w2/(1-f2)<=1]]
\end{lstlisting}

\paragraph{Code for Lemma~\ref{lem:resolve3+}} ~
\begin{lstlisting}
  Asym = Phi[F[1]] F[1] ((1-F[1])/(1-y[1])/Phi[y[1]]+Sum[F[j]/(1-y[j])/Phi[y[j]], {j,2,4}] + 4f[1]M Sum[F[k]/(1-f[k]),{k,2,4}]);
  Alpha = Asym/.{f[2]->(1-f[1])/3, f[3]->(1-f[1])/3,f[4]->(1-f[1])/3, y[3]->y[2], y[4]->y[2]};
  Resolve[Exists[{f[1],y[1],y[2]}, Alpha>963/1000 && 1/13<=f[1]<=1/2 && 0<=y[1]<=1/2 && 0<=y[2]<=1/2]]
\end{lstlisting}

\paragraph{Code for Lemma~\ref{lem:resolve2+1-}, case $f_1=\frac{1}{6}\tp{1-f_2}$} ~

\begin{lstlisting}
  F1 = F[1]/.{f[4]->(1-f[1]-f[2])/2, f[3]->(1-f[1]-f[2])/2,y[4]->y[3]};
  A = (1-f[1])(1-y[1])+(1-f[2])(1-y[2])+(1+f[1]+f[2])(1-y[3]);
  P1 = A(1-F1)(1/(1-y[1])/Phi[y[1]]-4M f[2]/(1-1/13));
  P2 = (1-f[2])/Phi[y[2]]+4f[1]M(1-y[2]);
  P3 = (1+f[1]+f[2])/Phi[y[3]]+8M(f[1]+f[2])(1-y[3]);
  Alpha1 = (1/A/(1/2-F1)*(P1+P2+P3))/.{f[1]->(1-f[2])/6};
  Resolve[Exists[{f[2],y[1],y[2],y[3]},Alpha1>9138/10000 && 0<=f[2]<=1/2 && 0<= y[1]<=1/2 && 0<=y[2]<=1/2 && 0<=y[3]<=1/2]]
\end{lstlisting}

\paragraph{Code for Lemma~\ref{lem:resolve2+1-}, case $f_1=\frac{1}{4}\tp{\frac{3}{4}-f_2}$} ~
\begin{lstlisting}
  F1 = F[1]/.{f[4]->(1-f[1]-f[2])/2, f[3]->(1-f[1]-f[2])/2, y[4]->y[3]};
  A =  (1-f[1])(1-y[1])+(1-f[2])(1-y[2])+(1+f[1]+f[2])(1-y[3]);
  P1 = A(1-F1)(1/(1-y[1])/Phi[y[1]]-4M f[2]/(1-1/13));
  P2 = (1-f[2])/Phi[y[2]]+4f[1]M(1-y[2]);
  P3 = (1+f[1]+f[2])/Phi[y[3]]+8M(f[1]+f[2])(1-y[3]);
  Alpha2 = (1/A/(1/2-F1)*(P1+P2+P3))/.{f[1]->(3/4-f[2])/4};
  Resolve[Exists[{f[2],y[1],y[2],y[3]}, Alpha2>9163/10000 && 0<=f[2]<=1/2 && 0<= y[1]<=1/2 && 0<=y[2]<=1/2 && 0<=y[3]<=1/2]]
\end{lstlisting}

\paragraph{Code for Lemma~\ref{lem:resolve2+1-}, case $f_1=\frac{1}{13}$} ~
\begin{lstlisting}
  F1 = F[1]/.{f[4]->(1-f[1]-f[2])/2, f[3]->(1-f[1]-f[2])/2, y[4]->y[3]};
  A = (1-f[1])(1-y[1])+(1-f[2])(1-y[2])+(1+f[1]+f[2])(1-y[3]);
  P1 = A(1-F1)(1/(1-y[1])/Phi[y[1]]-4M f[2]/(1-1/13));
  P2 = (1-f[2])/Phi[y[2]]+4f[1]M(1-y[2]);
  P3 = (1+f[1]+f[2])/Phi[y[3]]+8M(f[1]+f[2])(1-y[3]);
  Alpha3 = (1/A/(1/2-F1)*(P1+P2+P3))/.{f[1]->1/13};
  Resolve[Exists[{f[2],y[1],y[2],y[3]}, Alpha3>9102/10000 && 0<=f[2] && f[2]<=1/2 && 0<= y[1]<=1/2 && 0<=y[2]<=1/2 && 0<=y[3]<=1/2]]
\end{lstlisting}

In order to speed up above three code snippets for Lemma
\ref{lem:resolve2+1-}, we can simplify \texttt{Alpha}$_i$ to the form
$\frac{A}{B}$ where both $A\ge 0$ and $B\ge 0$ are polynomials. To
verify $\frac{A}{B}\le \alpha$, it is equivalent to verify $A-\alpha
B\le 0$, which can be done more efficiently by \emph{Mathematica}.

\paragraph{Code for Lemma~\ref{lem:jensen_mod}} ~
\begin{lstlisting}
G[w_,f_] := (1-f)/Phi[1-w/(1-f)] + M w/(1-f);
Resolve[Exists[{w1,w2,f1,f2}, (G[w1,f1]+G[w2,f2])>=1009/1000(2G[(w1+w2)/2, (f1+f2)/2]) && 0<=f1<=1/2 && 0<=f2<=1/2 && 1/2<=w1/(1-f1)<=1 && 1/2<=w2/(1-f2)<=1]]
Resolve[Exists[{w1,w2,f1,f2}, (G[w1,f1]/3 + 2G[w2,f2]/3)>=1009/1000(G[(w1+2w2)/3, (f1+2f2)/3]) && 0<=f1<=1/2 && 0<=f2<=1/2 && 1/2<=w1/(1-f1)<= 1 && 1/2<=w2/(1-f2)<=1]]
\end{lstlisting}

\paragraph{Code for Lemma~\ref{lem:resolve2+1-d1}} ~
\begin{lstlisting}
F1 = F[1]/.{f[4]->(1-f[1]-f[2])/2, f[3]->(1-f[1]-f[2])/2, y[4]->y[3]};
A = (1-f[1])(1-y[1])+(1-f[2])(1-y[2])+(1+f[1]+f[2])(1-y[3]);
P1 = A(1-F1)(1/(1-y[1])/Phi[y[1]]-2M f[2]/(1-f[1]));
P2 = (1-f[2])/Phi[y[2]]+ 2f[1] M(1-y[2]);
P3 = (1+f[1]+f[2])/Phi[y[3]]+4M(f[1]+f[2])(1-y[3]);
Alpha = (1/A/(1/2-F1)*(P1+P2+P3))/.{f[1]->0, y[1]->0};
Resolve[Exists[{f[2],y[2],y[3]}, Alpha>9231/10000 && 1/13<=f[2]<=1/2 && 0<=y[2]<=6/13 && 0<=y[3]<=6/13]]
\end{lstlisting} 

All the above verifications can be done within one hour on a laptop equipped with Intel i7-4700MQ CPU.


\section*{Acknowledgements}

We would like to thank Jingcheng Liu for insightful discussion.

\bibliographystyle{alpha} \bibliography{coloring}

\begin{thebibliography}{10}

\bibitem{bubley1997path}
Russ Bubley and Martin Dyer.
\newblock Path coupling: {A} technique for proving rapid mixing in {M}arkov
  chains.
\newblock In {\em Proceedings of the 38th Annual IEEE Symposium on Foundations
  of Computer Science (FOCS'97)}, pages 223--231. IEEE, 1997.

\bibitem{bubley1999approximately}
Russ Bubley, Martin Dyer, Catherine Greenhill, and Mark Jerrum.
\newblock On approximately counting colorings of small degree graphs.
\newblock {\em SIAM Journal on Computing}, 29(2):387--400, 1999.

\bibitem{dyer2006randomly}
Martin Dyer, Abraham~D Flaxman, Alan~M Frieze, and Eric Vigoda.
\newblock Randomly coloring sparse random graphs with fewer colors than the
  maximum degree.
\newblock {\em Random Structures \& Algorithms}, 29(4):450--465, 2006.

\bibitem{dyer2003randomly}
Martin Dyer and Alan Frieze.
\newblock Randomly coloring graphs with lower bounds on girth and maximum
  degree.
\newblock {\em Random Structures \& Algorithms}, 23(2):167--179, 2003.

\bibitem{dyer2013randomly}
Martin Dyer, Alan Frieze, Thomas~P Hayes, and Eric Vigoda.
\newblock Randomly coloring constant degree graphs.
\newblock {\em Random Structures \& Algorithms}, 43(2):181--200, 2013.

\bibitem{galanis2012inapproximability}
Andreas Galanis, Daniel {\v{S}}tefankovic, and Eric Vigoda.
\newblock Inapproximability of the partition function for the antiferromagnetic
  ising and hard-core models.
\newblock {\em arXiv preprint arXiv:1203.2226}, 2012.

\bibitem{gamarnik2012correlation}
David Gamarnik and Dmitriy Katz.
\newblock Correlation decay and deterministic {FPTAS} for counting colorings of
  a graph.
\newblock {\em Journal of Discrete Algorithms}, 12:29--47, 2012.

\bibitem{hayes2003randomly}
Thomas~P Hayes.
\newblock Randomly coloring graphs of girth at least five.
\newblock In {\em Proceedings of the 35th Annual ACM Symposium on Symposium on
  Theory of Computing (STOC'03)}, pages 269--278. ACM, 2003.

\bibitem{hayes2003non}
Thomas~P Hayes and Eric Vigoda.
\newblock A non-markovian coupling for randomly sampling colorings.
\newblock In {\em Proceedings of the 44th Annual IEEE Symposium on Foundations
  of Computer Science (FOCS'03)}, pages 618--627. IEEE, 2003.

\bibitem{hayes2006coupling}
Thomas~P Hayes and Eric Vigoda.
\newblock Coupling with the stationary distribution and improved sampling for
  colorings and independent sets.
\newblock {\em The Annals of Applied Probability}, 16(3):1297--1318, 2006.

\bibitem{jerrum1995very}
Mark Jerrum.
\newblock A very simple algorithm for estimating the number of k-colorings of a
  low-degree graph.
\newblock {\em Random Structures and Algorithms}, 7(2):157--166, 1995.

\bibitem{jonasson2002uniqueness}
Johan Jonasson.
\newblock Uniqueness of uniform random colorings of regular trees.
\newblock {\em Statistics \& Probability Letters}, 57(3):243--248, 2002.

\bibitem{li2012approximate}
Liang Li, Pinyan Lu, and Yitong Yin.
\newblock Approximate counting via correlation decay in spin systems.
\newblock In {\em Proceedings of the 23th Annual ACM-SIAM Symposium on Discrete
  Algorithms (SODA'12)}, pages 922--940. SIAM, 2012.

\bibitem{li2013correlation}
Liang Li, Pinyan Lu, and Yitong Yin.
\newblock Correlation decay up to uniqueness in spin systems.
\newblock In {\em Proceedings of the 24th Annual ACM-SIAM Symposium on Discrete
  Algorithms (SODA'13)}, pages 67--84. SIAM, 2013.

\bibitem{liu2015fptas}
Jingcheng Liu and Pinyan Lu.
\newblock {FPTAS} for counting monotone {CNF}.
\newblock In {\em Proceedings of the 26th Annual ACM-SIAM Symposium on Discrete
  Algorithms (SODA'15)}, pages 1531--1548. SIAM, 2015.

\bibitem{lu2013improved}
Pinyan Lu and Yitong Yin.
\newblock Improved {FPTAS} for multi-spin systems.
\newblock In {\em Proceedings of APPROX-RANDOM}, pages 639--654. Springer,
  2013.

\bibitem{molloy2004glauber}
Michael Molloy.
\newblock The {G}lauber dynamics on colorings of a graph with high girth and
  maximum degree.
\newblock {\em SIAM Journal on Computing}, 33(3):721--737, 2004.

\bibitem{salas1997absence}
Jes{\'u}s Salas and Alan~D Sokal.
\newblock Absence of phase transition for antiferromagnetic potts models via
  the dobrushin uniqueness theorem.
\newblock {\em Journal of Statistical Physics}, 86(3-4):551--579, 1997.

\bibitem{sinclair2014approximation}
Alistair Sinclair, Piyush Srivastava, and Marc Thurley.
\newblock Approximation algorithms for two-state anti-ferromagnetic spin
  systems on bounded degree graphs.
\newblock {\em Journal of Statistical Physics}, 155(4):666--686, 2014.

\bibitem{sly2012computational}
Allan Sly and Nike Sun.
\newblock The computational hardness of counting in two-spin models on
  d-regular graphs.
\newblock In {\em Proceedings of the 53rd Annual IEEE Symposium on Foundations
  of Computer Science (FOCS'12)}, pages 361--369. IEEE, 2012.

\bibitem{vigoda2000improved}
Eric Vigoda.
\newblock Improved bounds for sampling colorings.
\newblock {\em Journal of Mathematical Physics}, 41(3):1555--1569, 2000.

\bibitem{weitz2006counting}
Dror Weitz.
\newblock Counting independent sets up to the tree threshold.
\newblock In {\em Proceedings of the 38th Annual ACM Symposium on Theory of
  Computing (STOC'06)}, pages 140--149. ACM, 2006.

\end{thebibliography}

\end{document}